%% file: main.tex
\documentclass[sigconf, nonacm]{acmart}
\usepackage{times}
\usepackage{graphicx}
\usepackage{balance}  
\usepackage{algorithm}
\usepackage[noend]{algpseudocode}
\usepackage{amsmath,amsfonts,amssymb}
\usepackage{booktabs}
\usepackage{xspace}
\usepackage{epsfig,subfigure}
\usepackage{epstopdf}
\usepackage{upgreek}
\usepackage[show]{chato-notes}
\usepackage{url}
\usepackage{multirow}

\newcommand\vldbdoi{XX.XX/XXX.XX}
\newcommand\vldbpages{XXX-XXX}
\newcommand\vldbvolume{14}
\newcommand\vldbissue{1}
\newcommand\vldbyear{2020}
\newcommand\vldbauthors{\authors}
\newcommand\vldbtitle{\shorttitle}

\usepackage{eqparbox}

\newcommand{\kw}[1]{{\ensuremath {\mathsf{#1}}}\xspace}
\newcommand{\kwn}[1]{{\ensuremath {\mathsf{#1}}}}

\newcommand{\stitle}[1]{\vspace{1ex} \noindent{\bf #1}}
\long\def\comment#1{}

\newcommand{\black}[1]{\textcolor{black}{#1}}

\newcommand{\g}{\kw{g}}
\newcommand{\f}{\kw{f}}
\newcommand{\red}{\kw{IG}}
\newcommand{\rdc}{\kw{IG}}
\newcommand{\dist}{\kw{dist}}

\newcommand{\anc}{\kw{anc}}
\newcommand{\des}{\kw{des}}

\newcommand{\igs}{\kw{IGS}}
\newcommand{\mtigs}{\kw{kBM}-\kw{IGS}}

\newcommand{\reach}{\kw{reach}}
\newcommand{\yes}{\kw{Yes}}
\newcommand{\no}{\kw{No}}

\newcommand{\yescs}{\kw{Yes}-\kw{candidates}}
\newcommand{\yesc}{\kw{Yes}-\kw{candidate}}

\newcommand{\T}{\mathcal{H}}
\newcommand{\yset}{\mathcal{Y}}
\newcommand{\sset}{\mathcal{S}}
\newcommand{\can}{\mathcal{P}}

\newcommand{\eg}{\kw{Gain}}
\newcommand{\IGS}{\kw{IGS}}
\newcommand{\BinG}{\kw{BinG}}
\newcommand{\MTBinG}{\kw{BinG}}
\newcommand{\HGS}{\kw{HGS}}
\newcommand{\MTHGS}{\kw{HGS}}
\newcommand{\STIGS}{\kw{STBIS}}
\newcommand{\MTDiv}{\kw{kBM}-\kw{DP}}
\newcommand{\MTTopk}{\kw{kBM}-\kw{Topk}}
\newcommand{\MTfast}{\kw{kBM}-\kwn{DP}+\xspace}
\newcommand{\MTEP}{\kw{kBM}-\kw{EP}}
\newcommand{\MTND}{\kw{kBM}-\kw{Card}}
\newcommand{\Frame}{\kw{kBM}-\kw{IGS}}
\newcommand{\MTprob}{\kw{kBM}-\kw{IGS} \kw{problem}}
\newcommand{\single}{\kw{STBIS}}

\newcommand{\pr}{\kw{pr}}

\newcommand{\Syes}{\mathcal{Y}}

\newcommand{\pyes}{\kw{p_{Yes}}}
\newcommand{\pno}{\kw{p_{No}}}
\newcommand{\gyes}{\kw{g_{Yes}}}
\newcommand{\gno}{\kw{g_{No}}}

\newcommand{\cyes}{\hat{\can}}
\newcommand{\cno}{\bar{\can}}
\newcommand{\yyes}{\hat{\yset}}
\newcommand{\yno}{\bar{\yset}}

\newcommand{\child}{\kw{child}}
\newcommand{\parent}{\kw{par}}

\newcommand{\tar}{\mathcal{T}}
\newcommand{\DP}{\kw{DP}}
\newcommand{\UB}{\kw{UB}}

\newcommand{\calyes}{\kw{calg_{Yes}}}
\newcommand{\calno}{\kw{calg_{No}}}

\newcommand{\ignore}[1]{}
\newcommand{\nop}[1]{}
\newcommand{\eat}[1]{}
\newcommand{\eatSIGMOD}[1]{}

\newcommand{\revision}[1]{\black{#1}} 
\newcommand{\xlzhu}[1]{\textcolor{red}{#1}} 

\newtheorem{definition}{Definition} 
\newtheorem{example}{Example} 
\newtheorem{problem}{Problem} 
\newtheorem{theorem}{Theorem} 
\newtheorem{lemma}{Lemma}

\newtheorem{Rule}{Rule}

\newcommand{\squishlisttight}{
 \begin{list}{$\bullet$}
  { \setlength{\itemsep}{0pt}
    \setlength{\parsep}{0pt}
    \setlength{\topsep}{0pt}
    \setlength{\partopsep}{0pt}
    \setlength{\leftmargin}{2em}
    \setlength{\labelwidth}{1.5em}
    \setlength{\labelsep}{0.5em}
} }

\newcounter{qcounter}
\newcommand{\squishnumlist} {
\begin{list}{\arabic{qcounter}.~}{\usecounter{qcounter}} 
{  \setlength{\itemsep}{0pt}
    \setlength{\parsep}{0pt}
    \setlength{\topsep}{0pt}
    \setlength{\partopsep}{0pt}
    \setlength{\leftmargin}{2em}
    \setlength{\labelwidth}{1.5em}
    \setlength{\labelsep}{0.5em}
}}

\newcommand{\squishend}{
  \end{list}
}

\begin{document}


\title{Budget Constrained Interactive 
Search for Multiple Targets}



\author{
Xuliang Zhu$^{1}$, Xin Huang$^{1}$, Byron Choi$^{1}$, Jiaxin Jiang$^{1}$, Zhaonian Zou$^{2}$, Jianliang Xu$^{1}$}
\affiliation{%
 \institution{$^1$Hong Kong Baptist University, Hong Kong, China\\
             $^2$Harbin Institute of Technology, Harbin, China}
}
\affiliation{\{csxlzhu, xinhuang, bchoi, jxjian, xujl\}@comp.hkbu.edu.hk, znzou@hit.edu.cn}


\begin{abstract}
Interactive graph search leverages human intelligence to categorize target labels in a hierarchy, which is useful for image classification, product categorization, and  database search. However, many existing interactive graph search studies aim at identifying a single target optimally, and suffer from the limitations of 
 asking too many questions and not being able to handle multiple targets.

To address these two limitations, in this paper, we study a new problem of \underline{b}udget constrained 
\underline{i}nteractive \underline{g}raph \underline{s}earch for \underline{m}ultiple targets called \MTprob. Specifically, given a set of multiple targets $\tar$ in a hierarchy and two parameters $k$ and $b$, the goal is to identify a $k$-sized set of selections $\sset$, such that the closeness between selections $\sset$ and targets $\tar$ is as small as possible, by asking at most a budget of $b$ questions. We theoretically analyze the updating rules and design a penalty function to capture the closeness between selections and targets.
To tackle the \MTprob, we develop a novel framework to ask questions using the best vertex with the largest expected gain, which provides a balanced trade-off between target probability and benefit gain. 
Based on the \mtigs framework, we first propose an efficient algorithm \single to handle the \kw{Single}\kw{Target} problem, which is a special case of \mtigs. Then, we propose a dynamic programming based method \MTDiv to tackle the \kw{Multiple}\kw{Targets} problem. To further improve efficiency, we propose two heuristic but efficient algorithms, \MTTopk and \MTfast. \MTTopk develops a variant gain function and selects the top-$k$ vertices independently. \MTfast uses an upper bound of gains and prunes disqualified vertices to save computations. 
Experiments on large real-world datasets with ground-truth targets verify both the effectiveness and efficiency of our proposed algorithms. 
\end{abstract}

\maketitle

\begingroup\small\noindent\raggedright\textbf{PVLDB Reference Format:}\\
\vldbauthors. \vldbtitle. PVLDB, \vldbvolume(\vldbissue): \vldbpages, \vldbyear.\\
\href{https://doi.org/\vldbdoi}{doi:\vldbdoi}
\endgroup
\begingroup
\renewcommand\thefootnote{}\footnote{\noindent
This work is licensed under the Creative Commons BY-NC-ND 4.0 International License. Visit \url{https://creativecommons.org/licenses/by-nc-nd/4.0/} to view a copy of this license. For any use beyond those covered by this license, obtain permission by emailing \href{mailto:info@vldb.org}{info@vldb.org}. Copyright is held by the owner/author(s). Publication rights licensed to the VLDB Endowment. \\
\raggedright Proceedings of the VLDB Endowment, Vol. \vldbvolume, No. \vldbissue\ %
ISSN 2150-8097. \\
\href{https://doi.org/\vldbdoi}{doi:\vldbdoi} \\
}\addtocounter{footnote}{-1}\endgroup

\input{tex/intro}
\input{tex/relate}
\input{tex/problem}
\input{tex/analysis}
\input{tex/single}
\input{tex/multi}
\input{tex/exp}

\section{Conclusion and Future Work}\label{sec.conclusion}

In this paper, we study the problem of \mtigs to identify multiple targets in a hierarchy via a constrained budget of interactive questions. To effectively tackle the problem, we propose a novel \mtigs framework to select the vertex with the maximum expected gain to ask question. On the basis of the \mtigs framework, we develop \single algorithm to identify a single target and a dynamic programming based method \MTDiv to identify multiple targets. To accelerate the efficiency, we propose two heuristic algorithms \MTTopk and \MTfast to ask question on the vertex with the best alternative gain. Extensive experiments validate the effectiveness and efficiency of our proposed algorithms.
\revision{
This paper also 
opens up interesting questions. One challenging direction is how to control the quality of target selections, given that crowd workers may give wrong answers.  
It takes non-trivial efforts to identify wrong answers in various possible forms and make a balanced budget cost between \emph{trusting the given answers by asking further questions} and 
\emph{suspecting the given answers by repeating questions to double-check} in a principled way. 
}


\bibliographystyle{abbrv}
\bibliography{DAG}

\end{document}

%% file: tex/intro.tex
\section{Introduction}\label{sec.intro}

\begin{figure}[t]
\vspace{-0.2cm}
\centering
{
\subfigure[Hierarchical ImageNet Tree]{
\label{fig.hierarch.example}
\includegraphics[width=0.65\linewidth]{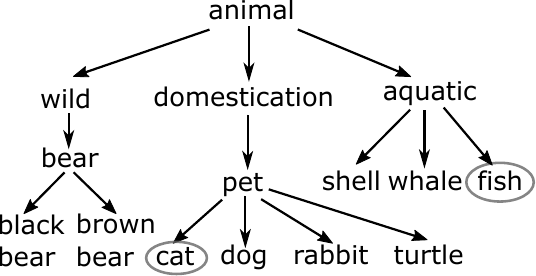} }
\subfigure[Image]{
\label{fig.fish}
\includegraphics[width=0.33\linewidth]{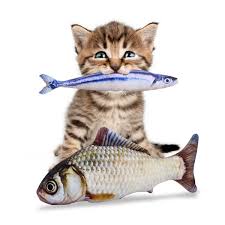} }
}
\vspace{-0.5cm}
\caption{An example of a hierarchical directed tree in ImageNet and an uncategorized image. The target labels of image in Figure~\ref{fig.fish} are $\tar=$\{``cat'', ``fish''\} in Figure~\ref{fig.hierarch.example}.}
\label{fig.intro.example}
\vspace{-0.5cm}
\end{figure}

Crowdsourcing, such as Amazon's Mechanical Turk~\cite{mturk} and CrowdFlower~\cite{crowdflower}, allows organizations to design human-aided services in which humans can help solve tasks and get rewards. 
In real applications, many tasks such as object categorization~\cite{parameswaran2011human}, entity resolution~\cite{vesdapunt2014crowdsourcing, wang2012crowder}, filtering noisy data~\cite{karger2011iterative, parameswaran2012crowdscreen}, ranking~\cite{marcus2012counting}, and labeling~\cite{barr2006ai}, are complex and difficult to resolve algorithmically.
With regard to human-aided object categorization, the 
graph search problem concerns leveraging human intelligence to categorize the target labels of a given object in a label hierarchy, which has a wide range of applications including image classification~\cite{li2020efficient,deng2009imagenet}, product categorization~\cite{ni2019justifying}, and relational database search~\cite{tao2019interactive}.

Recently, Tao et al.~\cite{tao2019interactive} investigated the problem of \emph{interactive} graph search (\igs) to locate \emph{one unique target vertex} in a hierarchy $\T$, with as few questions as possible. For example, Figure~\ref{fig.intro.example}(a) shows a hierarchical tree with several labeled vertices. A directed edge from one vertex to another  represents the concept-instance relationship, e.g., ``pet'' is a general concept of four instances ``cat'', ``dog'', ``rabbit'', and ``turtle''. Note that the target is unknown in advance. To identify the target, interaction is allowed to iteratively ask questions using the vertices in the hierarchy, e.g., ``Is this a wild?'', ``Is this a pet?''. Assuming that the target is ``pet'' in Figure~\ref{fig.intro.example}(a), we need to ask at least five questions of ``Is this $x$?'', where $x \in $ \{``pet'', ``cat'', ``dog'', ``rabbit'',  ``turtle''\}, to get the answers \{$\yes$, $\no$, $\no$, $\no$, $\no$\} and then determine the exact target of ``pet''. 
Effective algorithms with theoretical guarantee are proposed for finding the exact target using at most $\lceil \log_{2}{h} \rceil(1 + \lfloor \log_{2}{n} \rfloor) + (d - 1) \cdot \lceil \log_{d}{n} \rceil$ questions, where $n$, $d$, $h$ are respectively the number of vertices, the maximum out-degree, and the hierarchy height  in $\T$. However, two issues remain open:
\squishlisttight
\item \underline{Finding nearly-optimal targets using a constrained 
budget}. 
IGS \cite{tao2019interactive} may incur a high  cost to identify the exact target. 
It does not limit the number of questions that can be asked. 
In the worst case, the proposed algorithm asks $(d - 1) \cdot \lfloor \log_{d}{n} \rfloor$ questions to optimally identify the target. In real hierarchy datasets, the out-degree $d$ could be large, e.g., ImageNet~\cite{deng2009imagenet} has $d = 391$ and $n=74,401$. Thus, users may need to answer $782$ questions, which is not very practical. 
Even worse, asking questions is potentially costly~\cite{parameswaran2011human}, which motivates the problem of budget constrained \IGS  to bound the total cost. 
\item \underline{Identifying multiple targets}. Existing studies on the \IGS problem~\cite{tao2019interactive,li2020efficient} only consider a single target, where the answer has one and only one target. However, in real applications of object categorization, an object may have multiple labels. Even worse, it is difficult to determine in advance how many labels the object may have. For example, Figure~\ref{fig.fish} shows an uncategorized image object. 
Both ``cat'' and ``fish'' are suitable to label the object, but either one alone is not good enough. 
\end{list}

To address the above issues
, in this paper, we propose a new \mtigs problem of interactive graph search for identifying multiple targets $\tar$ using a constrained budget to ask at most $b$ questions. Specifically, in each round, our \mtigs scheme asks a question in the form ``Given a query vertex $q$ in tree $\T$, can vertex $q$ reach one of targets in $\tar$?'' and receives the answer from human-assisted interactions. On the basis of the previous answers, the \mtigs scheme determines the next question to ask. 
Finally, it selects a set of vertices 
to represent the targets 
after $b$ questions are answered.

However, it is significantly challenging to identify the most suitable selections in the \mtigs problem, for the following reasons. 
First, the number of targets is unknown in advance. Given an uncategorized object, it may have one or more ground-truth labels. Second, in the worst case, a total of $O(n)$ questions is needed to find the exact answers regarding targets, which makes the selection of $b$ questions difficult. Third, since the targets are unknown, another challenge is how to measure the goodness of a solution, i.e., the closeness between selections and targets.

In light of the above, the objective of our problem is formulated as finding a $k$-sized selection set of vertices to approach the targets as close as possible, w.r.t. an input number of $k$ and a budget of $b$ questions. 
For example, consider the hierarchy and the uncategorized image object in Figure~\ref{fig.intro.example}. Assume that $k=2$ and $b=2$. We ask two questions of ``Is this $x$?''\footnote{\scriptsize{The question is equivalent to a search question in the form ``Given a query vertex $q$ in tree $\T$, can vertex $q$ reach one of targets in $\tar$?''}}, where $x$ is ``pet'' and ``fish'', respectively, and both answers are $\yes$. After that, we cannot ask any more questions to verify the other four specified pets, i.e., ``cat'', ``dog'', ``rabbit'', and ``turtle''. Thus, we select ``pet'' and ``fish'' as the solution. Assume that the targets are ``cat'' and ``fish''. 
It can be seen that the selections of  ``pet'' and ``fish'' are close to the targets, since ``pet'' is a generalization of ``cat''. On the other hand, ``animal'' is also a good label, but it is far from ``cat'' and worse than our selection ``pet''. \revision{More real-world applications of our problem are presented in Section~\ref{sec.application}. Moreover, a detailed comparison between the state-of-the-art methods~\cite{parameswaran2011human,tao2019interactive,li2020efficient}  and our model is provided in Section~\ref{sec.relate}.}

To tackle the \mtigs problem, we propose a novel \mtigs framework, which uses a greedy strategy to ask the best question with the largest expected gain at each round. Specifically, vertices have different probabilities to be targets and may get \yes/\no answers for questions asked. In general, a vertex at the top level of the hierarchy has a high probability of getting a $\yes$ answer. However, the benefit of getting a $\yes$ answer can be less than a \no answer, which implies that none of descendants are targets. 
Therefore, we propose an expected gain to trade-off the target probability and benefit gain. 
Thus, the \mtigs framework can find the vertex with the largest expected gain to ask the next question.
On the basis of the \mtigs framework, we first propose an efficient algorithm STIGS to solve the \kw{SingleTarget} \kw{problem}
, which is a special case of \mtigs with $|\tar|=1$. 
It can update the gains of all vertices in $O(n)$ for each question. Different from the \kw{SingleTarget} \kw{problem}, it is difficult to calculate the gains in the \kw{Multiple}\kw{Targets} \kw{problem}. 
We then develop a \MTDiv method to calculate the optimal penalty between the $k$-sized selections and potential targets. The algorithm takes $O(bnh^2dk^2)$ time. To further improve efficiency, we propose two heuristic but efficient algorithms, \MTTopk and \MTfast. The first method, \MTTopk, uses an independent penalty function, which avoids costly enumerations to find the optimal result. The second method, \MTfast, improves \MTDiv by invoking an upper bound to quickly identify the best vertex to ask the next question before all expected gains are updated. 
To summarize, we make the following contributions: 
\squishlisttight
\item 
We propose a new \MTprob of budget constrained interactive graph search for identifying multiple targets in a hierarchical tree. We raise the problem of finding the $k$-sized selections close to multiple targets using a fixed number of $b$ questions, and formally design a penalty function to measure the closeness between selections and targets (Section~\ref{sec.problem}).
\item We give theoretical analysis of potential targets and \yesc, which offers useful updating rules to prune disqualified candidates. Moreover, for a specific question, we discuss the probabilities of \yes and \no answers, and analyze the expected gain of asking this question. On the basis of the updating rules and expected gains, we propose a novel \mtigs framework to tackle the \MTprob by asking $b$ good questions (Section~\ref{sec.analysis}).

\item We investigate and tackle one instance of \MTprob, the \kw{SingleTarget} \kw{problem}, where the target involves a single answer. On the basis of the \mtigs framework, we derive new updating rules and propose a greedy algorithm \STIGS. We further develop a DFS technique to accelerate \STIGS 
(Section~\ref{sec.single}).
\item We propose three efficient algorithms for identifying multiple targets based on the \mtigs framework. First, we propose a dynamic programming based technique \MTDiv, which computes the gain scores for all vertices optimally. To further improve  efficiency, we propose two fast algorithms, \MTTopk and \MTfast, based on an alternative function of independent gain score and a pruning upper bound for lazy gain computation, respectively (Section~\ref{sec.multi}).
\item We conduct extensive experiments on real-world datasets with \emph{ground-truth multiple targets} to validate the efficiency and effectiveness of our proposed framework and algorithms (Section~\ref{sec.exp}).
\end{list}


%% file: tex/relate.tex
\section{Related Work}\label{sec.relate}

\begin{table*}[t]
\small
\vspace{-0.3cm}
\caption{Comparison with relevant studies IGS~\cite{tao2019interactive}, BinG~\cite{li2020efficient}, and HGS~\cite{parameswaran2011human}. Here, $b$ is the budget of questions, and $n$, $d$, $h$ are respectively the number of vertices, the maximum out-degree, and the height in the hierarchy.}\label{table.relate}
\vspace{-0.4cm}
\centering
\scalebox{0.8}{
\begin{tabular}{||c||c||c||c|c||c|c||}
\toprule
Method& Interactive& Targets& Budget& Questions (Worst Case)&Time (Each Question)& Time (Total) \\
\midrule
IGS~\cite{tao2019interactive}& $\checkmark$& Single& $\times$& $\lceil \log_{2}{h} \rceil(1 + \lfloor \log_{2}{n} \rfloor) + (d - 1) \cdot \lceil \log_{d}{n} \rceil$& $O(1)$& $O(n\log{n})$\\
\hline
BinG~\cite{li2020efficient}& $\checkmark$& Single& $\times$& $n - 1$ & $O(n)$& $O(n^2)$\\
\hline
\multirow{2}*{HGS~\cite{parameswaran2011human}}& $\times$& Single& $\checkmark$& $b$& /& $O(n\log{n})$\\
\cline{2-7}
~& $\times$& Multiple& $\checkmark$& $b$& /& $O(b^2n^6)$\\
\hline
\multirow{2}*{\mtigs}& $\checkmark$& Single& $\checkmark$& $b$& $O(n)$& $O(bn)$\\
\cline{2-7}
~& $\checkmark$& Multiple& $\checkmark$& $b$& $O(nh^2dk^2)$& $O(bnh^2dk^2)$\\
\bottomrule
\end{tabular}
}
\vspace{-0.3cm}
\end{table*}

Our work is related to human-assisted data processing tasks~\cite{parameswaran2012crowdscreen, parameswaran2014datasift, li2018approximate, karger2011human, vesdapunt2014crowdsourcing, wang2013leveraging, whang2013question} and object categorization problems~\cite{chakrabarti2004automatic, gao2020channel, liu2020interactive, tang2020multi, zhou2015unsupervised, guo2019dual, karlinsky2017fine}. \revision{Table~\ref{table.relate} shows a detailed comparison of the three most relevant studies, IGS~\cite{tao2019interactive}, BinG~\cite{li2020efficient}, HGS~\cite{parameswaran2011human}, and our \mtigs.}
 Tao et al.~\cite{tao2019interactive} propose an \revision{\emph{interactive graph search} (IGS)} method for identifying a single target in a directed hierarchy. The general idea is to apply heavy-path decomposition to produce a balance representation of hierarchy and tackle the problem by binary searches. Li et al.~\cite{li2020efficient} model the single target problem as a decision tree construction problem. They propose a greedy based method for interactive graph search \revision{(denoted as BinG)}, which improves the performance of \IGS~\cite{tao2019interactive}. Both studies consider a single target and find the exact result using an unlimited budget of questions. Consider the example shown in Figure~\ref{fig.worst_case}. Assume that the target is $r$. Both \IGS~\cite{tao2019interactive} and \BinG~\cite{li2020efficient} would ask all its children to determine whether $r$ is the target, which takes $n - 1$ questions.
 Different from these two studies~\cite{tao2019interactive,li2020efficient}, our proposed framework can tackle both \kw{Single}\kw{Target} and \kw{Multiple}\kw{Targets} problems and select the representative targets within a bounded budget.
 \begin{figure}[h!]
\vspace{-0.5cm}
\centering
{
\subfigure{
\includegraphics[width=0.35\linewidth]{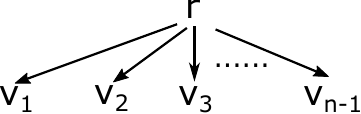} }
}
\vspace{-0.4cm}
\caption{A hierarchy has $n$ vertices with the target $r$.}
\label{fig.worst_case}
\vspace{-0.5cm}
\end{figure}

\revision{
Parameswaran et al.~\cite{parameswaran2011human} also investigate both \kw{Single}\kw{Target} and \kw{Multiple}\kw{Targets} problems with bounded budgets and propose \HGS to find multiple targets using $b$ questions.
However, the novelty of our problem is the consideration of the \emph{interactive} setting, which enables \emph{dynamic} algorithm designs and brings \emph{significant} performance benefits. 
First, the \HGS scheme is a \emph{non-interactive} algorithm that asks all $b$ questions in \emph{one go} and then, based on the workers' answers, does the best to figure out where the targets are. Its objective is to choose the $b$ questions wisely to minimize the size of the candidate set. As a result, it may not be able to find the candidates close to the targets. 
In contrast, our proposed approach leverages the answers of the previous $l$ questions ($1 \leq  l \leq b-1$) to \emph{dynamically} determine the $(l+1)$-th question. Such interaction allows our algorithm to quickly narrow down the search space and efficiently guide the search towards the targets. 
Second, the \HGS algorithm divides the whole hierarchy into $b$ subtrees and asks a question on each root of the $b$ subtrees. It has an extremely time complexity of $O(b^2n^6)$~\cite{parameswaran2011human}. Different from \HGS, our dynamic approach works by asking one question each time on a single vertex that achieves the largest expected gain based on the previous answers. In other words, our approach is a greedy algorithm that runs $b$ times to identify $b$ questions in $O(bnh^2dk^2)$ time. 
It is more effective and efficient than \HGS, as will be validated by the experiments in Section~\ref{sec.exp}. 
}

Furthermore, there exist a large number of studies on hierarchy construction using crowdsourcing techniques~\cite{sun2015building, chilton2013cascade, bragg2013crowdsourcing}. Sun et al.~\cite{sun2015building} build up a crowdsourcing system to construct a hierarchy by maintaining distributions over possible hierarchies. Chilton et al.~\cite{chilton2013cascade} propose an automated workflow Cascade to create a taxonomy from the collective efforts of crowd workers. Bragg et al.~\cite{bragg2013crowdsourcing} present an improved workflow DELUGE which uses less crowd labor to generate taxonomies. All these studies focus on how to build high quality category hierarchies. Orthogonal to these studies, our work utilizes the hierarchies and crowdsourcing ideas to categorize  objects in an interactive way w.r.t. a given budget of questions. 

%% file: tex/problem.tex
\section{Preliminaries}\label{sec.problem}
In this section, we present definitions and formulate our problem.

\subsection{Hierarchical Tree}

Let $\T=(V, E)$ be a directed hierarchical tree rooted at $r$ with a set $V$ of vertices and a set $E$ of directed edges, where the root $r\in V$ and the edge set $E=\{\langle v, u\rangle: v \text{ is the parent  of } u \}$.  Let the size of vertex set be $n = |V|$ and the height of $\T$ be $h$. 
For a vertex $v\in V$, we denote its children of $v$ by $\child(v) = \{u: \langle v, u\rangle \in E\}$ and its unique parent by $\parent(v)$ where $\langle \parent(v), v\rangle \in E$. 
A vertex $v$ with no children is called leaf, i.e., $\child(v)=\emptyset$. 

Given two vertices $u$ and $v$, we say that $u$ can reach $v$ (denoted as $u \rightarrow v$), if and only if there exists a directed path from $u$ to $v$ in $\T$. If $u$ cannot reach $v$, we use $u \nrightarrow v$ to represent it. Note that $v \rightarrow v$ and  $r\rightarrow v$ for any vertex $v\in V$. 
Moreover, the distance from $u$ to $v$ is denoted by $\dist \langle u, v\rangle$, as the length of the shortest path from $u$ to $v$ in $\T$. If $u \rightarrow v$, the distance $\dist \langle u, v\rangle$ is the height difference between two vertices $u$ and $v$ in $\T$.  Note that $\dist\langle v, v\rangle=0$  and $\dist \langle u, v\rangle$ $= +\infty$ if $u \nrightarrow v$.
In addition, the ancestors of a vertex $v$, denoted by $\anc(v)$, are defined to be the set of vertices that can reach $v$ in $\T$, i.e., $\anc(v)= $ $ \{u \in V: u \rightarrow v\}$. 
Similarly, the descendants of a vertex $v$, denoted by $\des(v)$, are defined to be the set of vertices that are reachable from $v$ in $G$, i.e., $\des(v)= $ $ \{u \in V: v \rightarrow u\}$. \revision{Table~\ref{tab.notat} lists the frequently used notations in the paper.}  


\begin{figure}[t]
\centering
{
\includegraphics[width=0.25\linewidth]{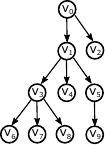}
}
\vspace{-0.3cm}
\caption{The hierarchical tree used in the running example.}
\label{example.prob}
\vspace{-0.6cm}
\end{figure}

\begin{example}
Figure \ref{example.prob} shows an example of hierarchical tree rooted by $v_0$. The children of $v_3$ are $\{v_6, v_7, v_8\}$ and its parent is $v_1$. We have $\anc(v_3) = \{v_0, v_1, v_3\}$ and $\des(v_3) = \{v_3, v_6, v_7, v_8\}$. The distance $\dist\langle v_0, v_3\rangle = 2$, $\dist\langle v_3, v_3\rangle = 0$ and $\dist\langle v_2, v_3\rangle = +\infty$.
\end{example}

\subsection{kBM-IGS Interactive Scheme}
In the following, we introduce the scheme of budget-based interactive graph search for identifying multiple targets. In contrast to IGS~\cite{tao2019interactive} and BinG~\cite{li2020efficient}, our \mtigs has new rules and features for asking a limited number of questions.
The interactive scheme 
has four 
components: \emph{targets}, \emph{questions}, \emph{\yescs}, and \emph{selections}.

\stitle{Targets.} The targets are a set of vertices in tree $\T$, denoted as $\tar \subseteq V$. The goal of our \mtigs is finding the target set $\tar$, which needs to be identified through a few rounds of question asking. 
The targets $\tar$ can be chosen arbitrarily from $V$. In other words,  $\tar$ may be any hidden vertices 
 distributed over the whole tree $\T$. 

The target set $\tar$ has two characteristics: \emph{variant cardinality} and \emph{target independence}. 
First, in terms of target cardinality, 
we categorize $\tar$ into two types, \kw{SingleTarget} and \kw{MultipleTargets}, following~\cite{parameswaran2011human}. 
If the size of $\tar$ is known and $|\tar|=1$, we call it \kw{Single Target}.
If the size of $\tar$ is unknown and variant with $|\tar|\geq 1$, we call it \kw{Multiple Targets}, which does not constrain the size of $\tar$. 
On the other hand, 
the target set $\tar$ must satisfy the property of target independence, that is, any two vertices in $\tar$ are not related~\cite{parameswaran2011human}:
\vspace{-0.08cm}
$$\forall v, u\in \tar, \text{ if } v\neq u   \text{, then } v\nrightarrow u.$$
Consider the example shown in Figure~\ref{fig.hierarch.example}, where we assume the target set is $\tar=$\{``cat'', ``fish''\}. Although ``pet'' is also a correct terminology to represent the object in Figure~\ref{fig.fish}, it is not suitable to be added into $\tar$, as it would violate the property of target independence as ``pet'' reaches ``cat'' in Figure~\ref{fig.hierarch.example}. Actually, ``cat'' is a more precise label to describe the object than ``pet'' in this example. 

\stitle{Questions.} To identify targets, one can ask a search question in the form ``Given a query vertex $q$ in tree $\T$, can vertex $q$ reach one of targets in $\tar$?''. Formally, 

\begin{definition}[Questions]
\label{def.question}
Given a query vertex $q$ and targets $\tar$ in tree $\T$, the search question is defined as $\reach(q)$. 
The boolean answer of $\reach(q)$ is either \yes or \no. \\
If $\reach(q) = \yes$, then $\exists t\in \tar$ such that $q \rightarrow t$; 
\\Otherwise, $\reach(q) = \no$, i.e., $\forall t\in \tar$, $q \nrightarrow t$. 
\end{definition}

For example, 
in Figure~\ref{fig.hierarch.example}, the question ``Can the vertex labeled `bear' reach one of targets in $\tar$?'' will get the $\no$ answer and neither ``bear'', ``black bear'', nor ``brown bear'' will be the correct label. On the contrary, the question ``Can the vertex labeled `pet' reach one of targets in $\tar$?'' will get the $\yes$ answer.
One-shot question asking is limited to figure out where  the targets are in a large tree $\T$. 
One can interactively ask more questions to identify the targets accurately. 
However, in our \mtigs setting, 
we are given a budget $b$ for the number of questions that can be asked. This is because asking questions is usually costly in real applications, e.g., on Mechanical Turk~\cite{mturk}. It is also not practical to ask users numerous questions as they may not be willing to answer too many questions.  
After $b$ rounds of question-asking and answer-checking, 
one finally makes a decision to choose the answers to represent the targets.

\revision{
\begin{table}[t]
\caption{\revision{Frequently used notations}}\label{tab.notat}
\vspace{-0.4cm}
\centering
\small
\scalebox{0.85}{
\begin{tabular}{|l|l|}
\toprule
Notation &  Description \\
\midrule
$\anc(v)$/$\des(v)$ & the set of $v$'s ancestors/descendants\\
$\child(v)$ & the set of $v$'s children\\
$\dist\langle u, v \rangle$ & the distance from $u$ to $v$\\
$\tar$,  $\sset$, $\yset$, $\can$ & the targets, selections, \yescs, and potential targets\\
$\reach(q)$& the question whether $q$ can reach one of targets $\tar$\\
$\f(u, v)$&  the pair-wise penalty between $u$ and $v$\\
$\f(\sset, \tar)$& the set-wise penalty between $\sset$ and $\tar$\\
$\cyes_v$/$\cno_v$& the updated potential targets after question $\reach(v)=\yes$/$\no$\\
$\yyes_v$/$\yno_v$& the updated \yesc after question $\reach(v)=\yes$/$\no$\\
$\pyes(v)$/$\pno(v)$& the probability of question $\reach(v)=\yes$/$\no$\\
$\g(\yset, \can, k)$& the minimum penalty of $\f(\sset, \can)$ with selection $\sset \subseteq \yset$ \\
$\gyes(v)$/$\gno(v)$& the benefit gain of question $\reach(v)=\yes$/$\no$\\
$\eg(q)$& the expected gain of asking question $\reach(q)$\\
\bottomrule
\end{tabular}
}
\vspace{-0.5cm}
\end{table}
}

\stitle{Yes-candidates}. We say that a vertex $v$ is a \yesc for targets $\tar$ if and only if $\exists t\in \tar$ such that $v \rightarrow t$. Obviously, the root $r$ is always a \yesc for any targets $\tar$. We define the \yescs as the set of \yesc for targets $\tar$ as follows. 

\vspace{-0.2cm}
\begin{definition}[\yescs]
\label{def.yset}
Given the targets $\tar$ in tree $\T$, and several rounds of asking questions  
$Q = \{q_0, q_1,$ $ ..., q_l\}$ where $l$ is a positive integer and $q_0=r$, the \yescs are defined as 
\vspace{-0.1cm}
$$\yset = \bigcup_{q_i \in Q, \reach(q_i) = \yes}{\anc(q_i)}.$$
\vspace{-0.3cm}
\end{definition}



\begin{example}
Assume that the targets $\tar=$$\{v_2, v_8\}$ in Figure~\ref{example.prob} and the questions $Q = \{v_0, v_2, v_3, v_5\}$. The answers are $\reach(v_2) = \yes, \reach(v_3) = \yes$, and $\reach(v_5) = \no$, thus $\yset = \{v_0, v_1, v_2, v_3\}$. 
\end{example}
\vspace{-0.2cm}

\stitle{Selections}. The selections, denoted as $\sset$, are a subset of \yescs $\yset$, which are selected by the algorithms to match the targets $\tar$ as closely as possible. 
For example, in Figure~\ref{fig.intro.example}, assume that we have questioned ``pet'' and get the $\yes$ answer. We can select ``animal'', ``domestication'', or ``pet'' because they must be the correct label. 

Overall, the goal of \mtigs interactive scheme is to use a few questions to determine the selections $\sset \subseteq \yset$ to approach the targets $\tar$ as closely as possible. 


\subsection{Penalty between Selections and Targets}
Given a budget of questions that can be asked and an unknown number of targets, it is  challenging to determine the locations of targets in a large tree $\T$. Instead of giving a simple boolean result, we develop an evaluation metric to quantify the goodness of our selections.  In the following, we introduce another important feature of  \emph{penalty} in \mtigs. The penalty is an  evaluation metric defined on the basis of distance, which measures the closeness between $\sset$ and $\tar$.

\stitle{Pair-wise Penalty}. Assume that we use a vertex $v\in V$ to cover a given target $t\in \tar$. If $v=t$, the choice $v$ exactly identifies the target $t$. If $v\neq t$, it needs to give a penalty score for using $v$ to cover the target $t$. On the basis of the distance, we give the definition of pair-wise  penalty score as follows. 
\vspace{-0.1cm}
\begin{equation}\label{eq:dist}
\f\langle v, t\rangle =
\left\{
\begin{aligned}
\dist\langle v, t\rangle, \ \text{if}~ v\in \anc(t) \\
\dist\langle r, t\rangle, \  \text{if}~ v\notin \anc(t) \\
\end{aligned}
\right.
\vspace{-0.1cm}
\end{equation}

By the above definitions, we consider two cases: 1) $v\in \anc(t)$ and 2) $v\notin \anc(t)$. First, for $v\in \anc(t)$, indicating $v\rightarrow t$, the best selection $v$ should have $\dist\langle v, t\rangle =0$. The further the distance $\dist\langle v, t\rangle$, the larger the penalty. Second, for $v\notin \anc(t)$, indicating $v\nrightarrow t$, we give a full penalty of the largest distance between $r$ and $t$, for using $v$ to cover $t$, i.e., $\f\langle v, t\rangle = \dist\langle r, t\rangle$. The deeper the location of target $t$, the larger the penalty. As a result, if a target can be reached by our selections, we use the shortest distance to indicate its closeness. Otherwise, if a target is not reachable from our selections, we give a distance-based penalty. 

\stitle{Set-wise Penalty}. On the basis of the pair-wise penalty, we give the definitions of set-wise penalty distance below.






\begin{definition}[Penalty]
\label{def.penalty}
Given a set of targets $\tar$ and a set of selections $\sset$, the penalty of $\sset$ covering a target $t\in \tar$ is  defined as the minimum penalty of using a vertex $v\in \sset$ to cover $t$, denoted as
\begin{equation}\label{eq:dist}
\f(\sset, t)=\min_{v\in \sset}  \f\langle v, t\rangle =  \min_{v\in \sset\cup\{r\}}\dist\langle v, t\rangle.
\end{equation}
Moreover, the penalty of $\sset$ covering targets $\tar$ is defined as the total penalty sum of $\sset$ covering all targets $t\in \tar$, denoted by 
\begin{equation}\label{eq:penalty}
\f(\sset, \tar)=\sum_{t\in \tar} \f(\sset, t) = \sum_{t\in \tar} {\min_{v\in \sset\cup\{r\}}\dist\langle v, t\rangle}.
\end{equation}
\end{definition}





Obviously, if $\sset = \tar$, the penalty is $\f(\sset, \tar)= 0$. The smaller the penalty, the better the selections $\sset$. In Figure~\ref{example.prob}, assume that $\sset = \{v_2, v_3\}$ and  $\tar = \{v_2, v_5, v_8\}$, thus $\f \langle v_2, v_8 \rangle = 3$ and $\f \langle v_3, v_8 \rangle = 1$. The set-wise penalty $\f(\sset, v_8) = 1$, $\f(S, v_5) = 2$ and $\f(\sset, \tar) = 3$.

\subsection{Problem Formulation}
On the basis of the above definitions, we formulate the problem of \underline{b}udget constrained 
\underline{i}nteractive \underline{g}raph \underline{s}earch for \underline{m}ultiple targets (\mtigs). 

\begin{problem}[\MTprob]
\label{prob.mtigs}
Given a hierarchical directed tree $\T = (V, E)$ rooted at $r$, a target set $\tar \subseteq V$, a budget of $b\geq 1$ questions that can be asked, and a positive integer $k$, the problem is asking $b$ questions $Q=\{q_0, q_1,...,q_b\}$ \emph{one by one} to determine a non-empty set of selections $\sset^* \subseteq \yset$ such that $|\sset^*|\leq k$ and the distance $\f(\sset^*, \tar)$ is the smallest. Equivalently, 
\begin{align}
  \sset^* & = \arg\min_{\sset \subseteq \yset, |\sset|\leq k} \f(\sset, \tar) \nonumber \\
  \text{s.t., } \yset &=\bigcup_{q_i \in Q, \reach(q_i) = \yes}{\anc(q_i)}. \nonumber 
\end{align}
\end{problem} 
Note that the maximum number of selections $k$ where $k \geq |\sset|$, could be either larger or smaller than $|\tar|$ as we do not know the number of targets $\tar$ in most applications of \mtigs. 
For the example in Figure~\ref{fig.intro.example}, assume that we get $\yes$ answer for questions ``pet'', ``fish'' and $\no$ answer for questions ``wild'', ``shell'', and ``whale''. The best selections $\sset^*$ with $k = 2$ will be ``pet'' and ``fish''. 




\input{tex/application}

%% file: tex/application.tex
\subsection{Applications}
\label{sec.application}
\revision{We motivate the \mtigs problem with three 
applications. }

\stitle{\revision{Image categorization.}} 
\revision{New images (e.g., biomedical images, surveillance photos, and user-uploaded images in online social networks) are continuously being generated and need to be classified by humans to identify objects and labels~\cite{parameswaran2011human, tao2019interactive}. Our \mtigs scheme can leverage the crowd-aided intelligence to identify multiple objects in an image using a budget constrained interactive graph search.
First, an image may have \underline{multiple labels}, e.g., the image shown in Figure~\ref{fig.fish} has two labels ``cat'' and ``fish''. 
Second, 
answering a question involves certain communication, latency, and monetary costs. Given \underline{a limited budget for rewards}, it is necessary to constrain the total number of questions to be asked and select the most suitable labels to categorize the image. 
}

\stitle{\revision{Manual curation.}} \revision{
Manual curation extends an existing taxonomy (e.g., Yago, Wikipedia, and web of concepts) by adding new entities to the taxonomy~\cite{tao2019interactive}. Given a new entity, interactive graph search can be used to find the nodes to \emph{parent} the new entity via crowdsourcing.  
Different from the existing study~\cite{tao2019interactive}, 
our \mtigs problem allows an entity to be a child of \underline{multiple parents}, which is common in real-world applications, e.g., ``whale'' is a subclass of both ``mammal'' and ``aquatic animal''. Moreover, on the basis of economic considerations, we can use \underline{a limited budget} to constrain the total number of questions that can be asked. 
}

\stitle{\revision{Cold-start recommendation.}} \revision{
Due to the lack of users' preferences in cold-start recommendations, online platforms (e.g., Twitter, TikTok, and YouTube) can ask a few questions to identify users' interests and then offer personalized recommendations in a more effective way. Users' preferences may be diverse, which are usually not limited to a single interest, e.g., one may like ``traveling'', ``financial news'', ``movies'', and so on. To avoid users becoming bored due to being asked too many questions, our \mtigs scheme can ask \underline{only a small number of questions} using an interest hierarchy  and adjusts its asking strategy \underline{dynamically based on the previous answers}.}

\stitle{\revision{Remark}.} \revision{In practical crowdsourcing applications, while human mistakes are inevitable, they can be minimized or eliminated by adopting effective quality control measures such as expert review, majority voting, group consensus, and so on~\cite{daniel2018quality}. As validated in \cite{tao2019interactive}, the influence of such mistakes on the outcome of the graph search algorithms is negligible. Thus, as with the previous works~\cite{parameswaran2011human,tao2019interactive,li2020efficient}, we assume in our algorithm design that the workers always give correct answers. For those cases where human mistakes are not eliminated and the workers give wrong answers, we will assess the quality of our methods in Section~\ref{sec.exp}. 
}


%% file: tex/analysis.tex
\section{The Proposed Framework}\label{sec.analysis}
In this section, we analyze the properties of \mtigs problem and briefly introduce our algorithmic framework.

\subsection{Theoretical Analysis}

\begin{figure}[t]
\vspace{-0.6cm}
\centering
{
\subfigure{
\includegraphics[width=0.25\linewidth]{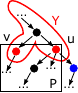} }
}
\vspace{-0.5cm}
\caption{\small{An example of \yescs and potential targets. The red vertices get $\yes$ answer, the blue vertices get $\no$ answers and the black vertices are not questioned. The red area is the \yescs and the black area is the potential targets.}}
\label{fig.cand.multi}
\vspace{-0.5cm}
\end{figure}

We first give new definitions of \emph{potential targets} 
and then analyze the relationships between \emph{questions} and \emph{potential targets}. 

\stitle{Potential targets.} We first define the potential targets, denoted by $\can$, as a candidate set of vertices that could be exact targets of $\tar$ where $\tar \subseteq \can$ . Obviously, if no question has been asked, every vertex $v\in V$ could be a potential target due to the limited prior information, i.e., $ \can = V$. However, as more questions are asked, the potential targets could decrease as some vertices may be pruned from $\can$ for violating the target constraints, no matter whether the question answer is $\yes$ or $\no$. We have the following lemmas.

\begin{lemma}
\label{lemma.noprune}
Given a vertex $q\in V$, if the question $\reach(q) = \no$, all vertices $u\in \des(q)$ are not targets, which should be pruned from potential targets, i.e., $\des(q) \cap  \can = \emptyset$. 
\end{lemma}
\vspace{-0.2cm}
\begin{proof}
First, for $\reach(q) = \no$, $q$ cannot reach any target $t\in \tar$. For each vertex $u\in \des(q)$, $u$ also cannot reach any target $t\in \tar$, $u \notin \tar$. Thus, $\des(q)\cap \tar = \emptyset$, and all vertices $\des(q)$ can be pruned from potential targets, denoted as $\des(q) \cap  \can = \emptyset$.
\end{proof}





\begin{lemma}
\label{lemma.yesprune}
Given a vertex $q\in V$, if the question $\reach(q) = \yes$, all vertices $u\in \anc(q)\setminus\{q\}$ are not targets, which should be pruned from potential targets, i.e., $\anc(q) \cap  \can = \{q\}$. 
\end{lemma}
\vspace{-0.2cm}
\begin{proof}
For $\reach(q) = \yes$, $\exists t \in \tar$ satisfies $q \rightarrow t$. On the basis of the independence property of targets, for other target $t^\prime \in \tar, t^\prime \neq t$, we can get $t^\prime \nrightarrow t$. Thus, for the vertex $u\in \anc(q)\setminus\{q\}$, $u \rightarrow t$, so $u \notin \tar$ and all vertices $u$ can be pruned from potential targets.
\end{proof}

For example, if we question the vertex $v_3$ in Figure~\ref{example.prob} and get the $\no$ answer, the vertices $v_3, v_6, v_7, v_8$ will be pruned from $\can$. Similarly, if we get the $\yes$ answer, the vertices $v_0, v_1$ will be pruned.


\stitle{Properties of Yes-candidates and potential targets.}
Next, we analyze the properties of the \yescs and potential targets.

\begin{lemma}\label{lemma.yescapcan}
$\yset \cap \can \neq \emptyset$ always holds.
\end{lemma}
\vspace{-0.2cm}
\begin{proof}
Given a series of all $Yes$ questions $Q = \{q_1, q_2, ... q_l\}$, which satisfy $\reach(q_i) = \yes$. Consider the vertex $q_i \in Q$ satisfies $q_i \nrightarrow q_j$ for all vertices $q_j \in Q \setminus \{q_i\}$, $q_i \in \can$. Furthermore, $q_i \in Q \subseteq \yset$. So $q_i \in \yset \cap \can$ and $\yset \cap \can \neq \emptyset$.
\end{proof}



\begin{lemma}
\label{lemma.target}
$\forall t \in \tar$ if and only if $\reach(t) = \yes$ and $\reach(u) = \no$ for all $u \in \child(t)$.
\end{lemma}
\vspace{-0.2cm}
\begin{proof}
Based on the definition of targets, $\forall t \in \tar$ satisfies $\reach(t) = \yes$ and $\reach(u) = \no$ for $\forall u \in \child(t)$. Moreover, if $\reach(t) = \yes$ then $\des(t) \cap \tar \neq \emptyset$. If $\reach(u) = \no$ for $\forall u \in \child(t)$, $(\des(t) \setminus \{t\}) \cap \tar = \emptyset$. So, $t \in \tar$.
\end{proof}

\begin{theorem}
\label{theorem.target}
If $\can \subseteq \yset$, the targets are exactly as $\tar = \can$. 
\end{theorem}
\vspace{-0.2cm}
\begin{proof}
For $\forall t \in \can$, $\reach(t) = \yes$. By Lemma~\ref{lemma.yesprune}, $\reach(v) = \no$ holds for $v\in \child(t)$. Moreover, by Lemma~\ref{lemma.target}, $\can \subseteq \tar$. As the definition of potential targets $\tar \subseteq \can$, thus $\can = \tar$. 
\end{proof}

Figure~\ref{fig.cand.multi} shows an example of \yescs and potential targets. 
$\can \cap \yset \neq \emptyset$. If the children of $u$ and $v$ are all questioned and get the $\no$ answer, $\can = \{u, v\} \subseteq \yset$ and the targets will be $\tar = \{u, v\}$. 


\subsection{kBM-IGS Framework}

In this section, we introduce a novel \Frame framework for identifying multiple targets via a series of $b$ interactive questions. The key idea is asking \emph{good questions} to reduce 
potential targets $\can$ and refine $\yset$ to be specified by Theorem~\ref{theorem.target}. 

\stitle{Motivations.} We use 
a toy example $\T$ in Figure~\ref{example.prob} to show the general ideas of our framework. 
Assume that $\can=V$ and $\yset = \{r\}$. 
First, we consider a vertex $v_1$ and ask the question $\reach(v_1)$. 
If $\reach(v_1) = \no$, all the descendants of $v_1$ can be pruned from $\can$, which achieves a considerable gain by reducing $|\can|$ from 10 to 2. 
But, if $\reach(v_1) = \yes$, $\can$ will only reduce the vertex $v_0$, which achieves a limited gain. 
Unfortunately, assuming that the targets are randomly distributed in $V$, $v_1$ has a low probability of getting $\reach(v_1) = \no$. This is because there exist 8 descendants of $v_1$, and if any vertex $u\in \des(v_1)$ is the target, $\reach(v_1) = \yes$ holds. 
Thus, $v_1$ may not be a good choice for questioning.
Second, we consider a leaf vertex $v_9$. If $\reach(v_9)=\yes$, we surely know that $v_9$ is one desired target, i.e., $v_9\in \tar$, which achieves lots of gains. But, if $\reach(v_9) = \no$, $\can$ will only reduce $v_9$, which achieves a limited gain. However, it has a high probability of getting $\reach(v_9) = \no$ and a low probability to $\reach(v_9) = \yes$. 
 We need to select good vertices by making a balanced trade-off between the probability and gains. 
To do so, the \Frame framework develops a ranking evaluation function for vertices, which is based on \emph{target probability} and \emph{gain score}.

\stitle{Target probability}. Naturally, the vertices at the top levels of $\T$ (e.g., root $r$ in Figure~\ref{example.prob}) have high probabilities of getting a $\yes$ answer and low probabilities to get a $\no$ answer. 
\revision{On the contrary, the vertices at the bottom levels of $\T$ (e.g., leaf $v_9$ in Figure~\ref{example.prob}) have high probabilities of getting a $\no$ answer and low probabilities to get a $\yes$ answer.} 
Therefore, for a vertex $v$, we denote the probabilities of $\reach(v)=\yes$ and $\reach(v)=\no$ respectively as $\pyes(v)$ and $\pno(v)$, which satisfy $\pyes(v) + \pno(v) = 1$. The specific calculations of $\pyes(v)$ and $\pno(v)$ are based on the  descendants $\des(v)$, which will be introduced in Sections~\ref{sec.single} and \ref{sec.multi}.  

\stitle{Gain score}. First, we define the potential penalty. Instead of using the targets $\tar$ as in Problem~\ref{prob.mtigs}, we define the potential penalty to measure the minimum distance between feasible selections $\sset$ and potential targets $\can$ as we do not know the exact $\tar$
, as follows.

\begin{definition}[Potential Penalty]
Given a set of potential targets $\can$, a set of \yescs $\yset$, and a number $k$, the potential penalty is denoted as $\g(\yset, \can, k) = \min_{\sset \subseteq \yset, |\sset| \leq k}{\f(\sset, \can)}$.
\end{definition}

The potential penalty $\g(\yset, \can, k)$ is to select the best $k$ vertices from the \yescs $\yset$ in order to identify the potential targets in $\can$. The less $|\can|$ and the closer $\sset$ to $\can$ is, the lower score $\g(\yset, \can, k)$ is, which is better. On the basis of potential penalties, we present the definition of gain score. For a given vertex $v\in V$ with existing $\can$ and $\yset$, we ask a new question $\reach(v)$, and present two gain scores for the different answers $\reach(v) = \yes$ and $\reach(v) = \no$ respectively, as follows.

\begin{definition}[Yes\&No Gains]
\label{def.gain}
The gain of $\reach(v)$ $=\yes$ is denoted as $ \gyes(v)$ $ = \g(\yset, \can, k) -$ $\g(\hat{\yset}_{v}, \hat{\can}_{v}, k)$ where $\hat{\yset}_{v}$, $\hat{\can}_{v}$ are the updated potential targets and \yescs after asking the question $\reach(v)=\yes$; Similarly, the gain of $\reach(v)=\no$ is denoted as $ \gno(v) = \g(\yset, \can, k) - $ $\g(\bar{\yset}_{v}, \bar{\can}_{v}, k)$ where $\bar{\yset}_{v}$, $\bar{\can}_{v}$ are the updated potential targets and \yescs after asking the question $\reach(v)=\no$.

\end{definition} 


On the basis of the target probabilities and gain scores, we define an integrated function of expected gain as follows.
\begin{definition}[Expected Gain]
\label{def.eg}
Given a vertex $v$ in $\T$, potential targets $\can$, \yescs $\yset$, and a number $k$, the expected gain of asking the question $\reach(v)$ is denoted as $$\eg(v) = \gyes(v) \cdot \pyes(v) + \gno(v) \cdot \pno(v).$$
\end{definition} 

The larger the expected gain of a vertex, the better the choice for questioning.

\begin{algorithm}[t]
  \small
  \caption{\Frame Framework}
  \label{algo:framework}
  \begin{algorithmic}[1]
    \Require A hierarchy tree $\T=(V, E)$, a budget $b$, a number $k$.
    \Ensure Selections $\sset$ with $|\sset|\leq k$.
    \State Let $\yset \leftarrow \{r\}$, $\can \leftarrow V$;
    \State Initialize the probability $\pr(v)$ for every vertex $v\in V$;
    \For {$i \leftarrow$ 1 to $b$}
      \For {$v\in \can \setminus \yset$}
        \State Calculate $\pyes(v), \pno(v), \gyes(v), \gno(v)$;
        \State $\eg(v) \leftarrow \gyes(v) \cdot \pyes(v) + \gno(v) \cdot \pno(v)$ by Def.~\ref{def.eg};
      \EndFor
      \State $q_i \leftarrow \arg\max_{v \in \can \setminus \yset} \eg(v)$;
      \State Ask the question $\reach(q_i)$;
      \If {$\reach(q_i) = \yes$}
        \State $\yset \leftarrow \yset \cup \anc(q_i)$ by Def.~\ref{def.question};
      \EndIf
      \State Update the potential candidates $\can$ and vertex probabilities $\pr$ accordingly if needed; 
      \If {$\can \subseteq \yset$} \textbf{break} by Theorem.~\ref{theorem.target};
      \EndIf
    \EndFor
    \State $\sset^* \leftarrow \arg \min_{\sset \subseteq \yset, |\sset| \leq k}{\f(\sset, \can)}$;
    \State \Return{$\sset^*$}; 
  \end{algorithmic}
\end{algorithm}

\stitle{Algorithm}. The algorithm of \Frame framework is outlined in Algorithm~\ref{algo:framework}. The general idea is to use a greedy strategy to select the vertex with the largest expected gain at each round of question-asking. The framework has an input of a hierarchy tree $\T$, a budget of $b$ questions that can be asked, and a number $k$. First, it initializes the \yescs $\yset$ as $\{r\}$ and the potential targets $\can$ as the whole vertex set $V$ (line 1). Note that if the vertices have no probabilities, we can set all vertices to have the same probability as $\frac{k}{n}$ (line 2). The algorithm then iteratively selects one best vertex $q_i\in \can \setminus \yset$ and asks the question $\reach(q_i)$ until the quota of $b$ questions is used up (lines 3-12). For each round, it calculates the target probabilities of $\pyes(v), \pno(v)$ and the \yes\&\no gains of $\gyes(v), \gno(v)$ for each vertex $v \in \can \setminus \yset$. Then, the expected gains of all vertices are computed (lines 4-6). The algorithm next finds the vertex $q_i$ with the largest expected gain and asks the question (lines 7-8). According to the answer, the \yescs, potential targets, and vertex probabilities are updated in accordance with the answer (lines 9-11). Finally, after $b$ questions or the identification of exact targets, the algorithm selects the best selection $\sset^* = \arg \min_{\sset \subseteq \yset, |\sset| \leq k}{\f(\sset, \can)}$ and return $\sset^* $ as the final selections (lines 13-14).

%% file: tex/single.tex
\section{Single Target Search}  \label{sec.single}
In this section, we investigate one special case of \mtigs problem, i.e., the \kw{SingleTarget} \kw{problem}, where $\tar$ has only a single target as $|\tar|=1$~\cite{li2020efficient, tao2019interactive}. On the basis of the \Frame framework, we develop a \single method to identify one vertex as $\sset$ by asking $b$ questions.


\subsection{Single Target Problem Analysis}
As an instance problem, the \kw{SingleTarget} \kw{problem} inherits all properties of \mtigs described in Section~\ref{sec.analysis} and enjoys its own properties. Assume that the initial $\can=V$ and $\yset=\{r\}$, and the target $\tar=\{t\}$. We can ask a question and interactively update $\can$ as $\can_{new}$ and $\yset$ as $\yset_{new}$ by obeying the two following rules.


\begin{figure}[t]
\vspace{-0.4cm}
\centering
{
\subfigure{
\includegraphics[width=0.25\linewidth]{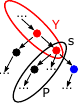} }
}
\vspace{-0.4cm}
\caption{An example of yes candidates and potential targets for single target.}
\label{fig.cand.single}
\vspace{-0.5cm}
\end{figure}


First, as $|\tar|=1$, for each question, the best strategy is to ask a vertex $v$ that is a potential target $v\in \can \setminus \yset$. Otherwise, we consider two cases. First, if we ask a vertex $v \in \yset$, the answer of $\reach(v)$ is always $\yes$; Second, if we ask a vertex $v \in V \setminus (\can \cup \yset)$, the answer of $\reach(v)$ is always \no. Thus, it achieves no benefit gain but wastes one question from the budget. 
Moreover, in contrast to Lemma~\ref{lemma.yesprune}, in \kw{SingleTarget} problem, more vertices can be pruned after $\yes$ answer as follows.

\begin{lemma}\label{lemma.single.prune}
For a vertex $q\in \can \setminus \yset$, if the question $\reach(q) = \yes$, none of the vertices $u \in \can\setminus\des(q)$ are potential targets and we update $\can_{new} = \des(q) \cap \can$. 
\end{lemma}
\vspace{-0.2cm}
\begin{proof}
Because $\reach(q) = \yes$, $\exists t \in \tar$ such that $t \in \des(q)$. Moreover, $|\tar| = 1$, $\tar = \{t\} \subseteq \des(q)$. Thus, $\tar \cap (V \setminus \des(q)) = \emptyset$ and all vertices $u \in V \setminus \des(q)$ should be pruned from $\can$.
\end{proof}

Second, the \yescs can be updated only when a question $\reach(v)=\yes$ by Def.~\ref{def.question}, where the updated \yescs $\yset_{new}= \yset\cup \anc(v)$. However, $\yset \subseteq \yset_{new} \subseteq \anc(t)$ always holds, i.e., all \yescs lie along the path from root $r$ to target $t$. The penalty function $\f(\sset, \tar)$ in Def~\ref{def.penalty} tells us that keeping one vertex $s\in \yset_{new}$ closest to $t$ is enough. In other words, it achieves the minimum penalty $\f(\sset, \tar)$ = $\f(\{s\}, \{t\})$. 
Thus, $\yset_{new}$ can be updated as $\yset_{new}=\{s\}$, where $s$ has the largest depth $\dist\langle r, s\rangle$ in $\T$ and the question $\reach(s)=\yes$.

\begin{lemma}
For $\yset=\{s\}$, both $\yset \cap \can = \{s\}$ and the penalty $\g(\yset, \can, 1) = \f(\{s\}, \can)$  hold.
\end{lemma}
\vspace{-0.2cm}
\begin{proof}
First, we prove $s \in \can$. \revision{Since $s \in \yset$, $s$ will not be pruned from $\no$ questions according to Lemma~\ref{lemma.noprune}.} Furthermore, as $s$ is the deepest vertex in $\yset$, $s$ will not be pruned from $\yes$ questions according to Lemma~\ref{lemma.single.prune}. So, $s \in \can$ and $\yset \cap \can = \{s\}$. Moreover, since $|\yset| = |\{s\}| = 1$, $\g(\yset, \can, 1) = \f(\{s\}, \can)$.
\end{proof}

On the basis of the above properties, we have two useful updating rules.

\begin{Rule}\label{rule.single.yes}
For a vertex $v\in \can$ with $\reach(v) = \yes$, we update the potential targets $\can_{new} = \des(v) \cap \can$ and the \yescs $\yset_{new}=\{v\}$. 
\end{Rule}
\vspace{-0.2cm}

\begin{Rule}\label{rule.single.no}
For a vertex $v\in \can$ with $\reach(v) = \no$,  we update the potential targets $\can_{new} = \can \setminus \des(v)$ and keep the \yescs unchanged  $\yset_{new}= \yset = \{s\}$. 
\end{Rule}

\subsection{The STBIS Algorithm}

In this section, we propose a greedy algorithm \single to find a single target. We begin with probability calculation. 

\stitle{Probability calculation.} Before asking any questions, each vertex has an equal probability of being the target. Thus, we let each vertex $u$ have a probability of $\pr(u) = \frac{1}{n}$ where $n=|V|$. Our proposed algorithm can be easily extended to other vertex probability distribution based on historical query logs as ~\cite{li2020efficient}. 
Therefore, for a vertex $v$, the target probability for each vertex $v$ follows $\pyes(v) = \sum_{u \in \des(v)}{\pr(u)}$ and the no probability follows $\pno(v) = 1 - \pyes(v)$. As more question answers are discovered, the vertex probabilities need to be updated accordingly. 
The updated probability after each question is calculated as:
\begin{equation} \label{eq.prob}
\pr(u) = 
\left\{
\begin{aligned}
&\pr(u) \cdot \frac{|\can|}{|\can_{new}|}, \ &u \in \can_{new}\\
&0, \ &u \notin \can_{new} \\
\end{aligned}
\right.
\end{equation}
where $\can_{new}$ is the potential targets after asking a question $\reach(q_i)$ where $1\leq i\leq b$. The general idea is to assign impossible targets with a probability value of zero and keep the sum probability equal to $1$.

\stitle{STBIS algorithm.} The detailed procedure of \single is outlined in Algorithm~\ref{algo:single}, which finds the vertices with the largest gain to ask interactive questions and finally identifies a selection to represent the target within $b$ questions. The algorithm first initializes the \yescs $\yset$ as a root $r$ and potential targets $\can = V$ (line 1), and uniformly assigns the vertex probability (line 2). Then, it 
calculates the $\yes$\&$\no$ probability (line 5), the $\yes$\&$\no$ gain scores (lines 6-9), and the expected gain $\eg(v)$ (line 10) for all potential targets $v\in \can$. Next, the algorithm chooses a vertex $q_i \in \can$ with the largest expected gain and ask question $\reach(q_i)$ (lines 11-12). If $\reach(q_i) =\yes$, it updates
$\can_{new} = \des(q_i)$ and $\yset=\{q_i\}$ by Rule~\ref{rule.single.yes} (lines 13-14); Otherwise, if $\reach(q_i) =\no$, it updates $\can_{new}= \can \setminus \des(q_i)$ by Rule~\ref{rule.single.no} (lines 15-16). It updates the probability using Eq.~\ref{eq.prob} (lines 17-18), and assign  $\can = \can_{new}$. Finally, the algorithm returns the vertex $s\in \yset$ as the selection (line 21) and terminates early if $\can = \yset$ by Theorem~\ref{theorem.target} (line 20).


\begin{algorithm}[t]
  \small
  \caption{\single}
  \label{algo:single}
  \begin{algorithmic}[1]
    \Require A hierarchy tree $\T=(V, E)$, root $r$, budget $b$, and $k=1$.
    \Ensure One selection $s$. 
    \State Let $\yset \leftarrow \{r\}$, $\can \leftarrow V$;
    \State Assign the probability $\pr(v)=1/n$ for $v\in V$;
    \For {$i \leftarrow$ 1 to $b$}
      \For {$v \in \can \setminus \yset$}
        \State 
        $\pyes(v) \leftarrow \sum_{u \in \des(v)}{\pr(u)}$, $\pno(v) \leftarrow 1 - \pyes(v)$;
        \State Calculate $\cyes_v$ as $\can_{new} \leftarrow \can \cap \des(v)$ by Rule~\ref{rule.single.yes};
        \State Update $\gyes(v) \leftarrow \f(\yset, \can) - \f(\{v\}, \can_{new})$;
        \State Calculate $\cno_v$ as $\can_{new} \leftarrow \can \setminus \des(v)$  by Rule~\ref{rule.single.no};
        \State Update $\gno(v) \leftarrow \f(\yset, \can) - \f(\yset, \can_{new})$;
        \State $\eg(v) \leftarrow \gyes(v) \cdot \pyes(v) + \gno(v) \cdot \pno(v)$;
      \EndFor
      \State $q_i \leftarrow \arg\max_{v \in \can \setminus \yset} \eg(v)$;
      \State Ask the question $\reach(q_i)$;
      \If {$\reach(q_i) = \yes$}
        \State $\can_{new} \leftarrow \can \cap \des(q_i)$; $\yset \leftarrow \{q_i\}$;
      \Else
        \State $\can_{new} \leftarrow \can \setminus \des(q_i)$; $\yset$ keeps unchanged;
      \EndIf
       \State Update $\pr(u) = 0$ for $u \in \can \setminus \can_{new}$;
       \State Update $\pr(u) = \pr(u) \cdot \frac{|\can|}{|\can_{new}|}$ for $u \in \can_{new}$;
      \State $\can \leftarrow \can_{new}$;
      \If {$\can = \yset$} \Return{ $s\in \yset$};
      \EndIf
    \EndFor
    \State \Return{$s\in \yset$}; 
  \end{algorithmic}
\end{algorithm}

\begin{algorithm}[t]
  \small
  \caption{DFS-Gain(u)}
  \label{algo:single.dfs}
  \begin{algorithmic}[1]
    \Require A subtree $T_u$, the root of subtree $u$.
    \Ensure $\eg(u)$.
    \State $\pyes(u) = \pr(u)$, $|\cyes_u| = 1$, $\f(\{u\}, \cyes_u) = 0$;
    \For{vertex $v \in \child(u)$}
      \If {$v \notin \can$} $\textbf{continue}$;
      \EndIf
      \State DFS-Gain(v);
      \State $\pyes(u) = \pyes(u) + \pyes(v)$;
      \State $|\cyes_u| = |\cyes_u| + |\cyes_v|$;
      \State $\f(\{u\}, \cyes_{u}) = \f(\{u\}, \cyes_{u}) + \f(\{v\}, \cyes_v)$;
    \EndFor
    \State $\pno(u) = 1 - \pyes(u)$;
    \State Calculate $\gyes(u), \gno(u), \f(\{u\}, \hat{\can}_{u})$ by Eq.~\ref{eq.fyes}, \ref{eq.gyes}, \ref{eq.gno};
    \State $\eg(u) \leftarrow \gyes(u) \cdot \pyes(u) + \gno(u) \cdot \pno(u)$;
    \State \Return{$\eg(u)$};
  \end{algorithmic}
\end{algorithm}

\begin{table}[t]
\vspace{-0.4cm}
\caption{The value of $\gyes$, $\gno$, $\pyes$, and $\pno$ of first question and the $\eg$ of two questions. Here, $b = 2$ and $\tar=\{v_5\}$.}
\vspace{-0.4cm}
\centering
\scalebox{0.75}{
\begin{tabular}{c|ccccccccc}
\toprule
Node& $v_{1}$& $v_{2}$& $v_{3}$& $v_{4}$& $v_{5}$ &$v_{6}$ &$v_{7}$& $v_8$& $v_9$\\
\midrule
$\gyes$& 9& 20& 17& 20& 19& 20& 20& 20& 20\\
$\gno$& 19& 1& 11& 2& 5& 3& 3& 3& 3\\
$\pyes$& 0.8& 0.1& 0.4& 0.1& 0.2& 0.1& 0.1& 0.1& 0.1\\
$\pno$& 0.2& 0.9& 0.6& 0.9& 0.8& 0.9& 0.9& 0.9& 0.9\\
$\eg^1$& 11& 2.9& \textbf{13.4}& 3.8& 7.8& 4.7& 4.7& 4.7& 4.7\\
$\eg^2$& \textbf{6}& 2.33& /& 3.17& \textbf{6}& /& /& /& 4\\
\bottomrule
\end{tabular}
}
\label{table:single}
\vspace{-0.6cm}
\end{table}

\begin{example}
Assume the target vertex is $v_5$ and the budget is $2$ in Figure~\ref{example.prob}. Table~\ref{table:single} shows the gains of each vertex. \revision{In the first round, the algorithm questions $v_3$ and gets the $\no$ answer.} Then, the vertices $v_3, v_6, v_7, v_8$ will be pruned. The vertices $v_1$ and $v_5$ get the maximum gain in the second round. If $v_5$ is selected, the penalty is $0$. If $v_1$ is selected, the penalty is $1$.
\end{example}
\vspace{-0.2cm}

\stitle{Complexity analysis and optimizations.} A straightforward implementation of \single in Algorithm~\ref{algo:single} takes $O(n)$ time to compute $\gyes(v)$, $\gno(v)$, $\pyes(v)$ and $\pno(v)$ for each vertex $v\in V$. Thus, the calculation of $\eg(v)$ for all vertices $v\in \can$ takes $O(n^2)$ time. However, these calculations have several repeated process, which can be avoided.

We consider the following deviations of $\yes$\&$\no$ gains. 

\begin{equation} \label{eq.fyes}
\begin{aligned}
\f(\{v\}, \cyes_v) 
&= \sum_{u \in \child(v)}{\f(\{u\}, \cyes_u}) + |\cyes_v| - 1
\end{aligned}
\end{equation}
\vspace{-0.2cm}
\begin{equation} \label{eq.gyes}
\begin{aligned}
\gyes(v) &= \f(\{s\}, \can) - \f(\{v\}, \cyes_v)\\
\end{aligned}
\end{equation}
\vspace{-0.2cm}
\begin{equation} \label{eq.gno}
\begin{aligned}
\gno(v) 
&= \f(\{v\}, \cyes_v) + \dist \langle s, v \rangle \cdot |\cyes_v|
\end{aligned}
\end{equation}

We propose a fast approach in Algorithm~\ref{algo:single.dfs} for computing expected gain $\eg(u)$ for all vertices $u$ based on Eqs.~\ref{eq.fyes},~\ref{eq.gyes} and ~\ref{eq.gno}. The algorithm can reduce the calculations to $O(n)$ time. Algorithm~\ref{algo:single.dfs} shows the details of DFS algorithm. First, the algorithm initializes the yes probability $\pyes$, the size of potential targets $|\cyes_u|$ and $\f(\{u\}, \cyes_u) $ (line 1). Next, it traverses the children of $u$ and updates $\pyes$, $|\cyes_u|$ and $\f(\{u\}, \cyes_u) $ (lines 2-7). Finally, the algorithm calculates $\pno$, $\gyes$, $\gno$ and returns them (lines 8-11). We can replace the procedure of lines 5-10 in Algorithm~\ref{algo:single} by invoking DFS-Gain(s) in Algorithm~\ref{algo:single.dfs}.

In summary, \single in Algorithm~\ref{algo:single} takes $O(n)$ time to generate a question for users to answer. The overall time complexity of \single takes $O(bn)$ time in $O(n)$ space for generating $b$ questions. 



%% file: tex/multi.tex
\section{Multiple Targets Search} 
\label{sec.multi}
In this section, we propose efficient 
algorithms for identifying $k$ selections for multiple targets. We first introduce the probability setting and updating rules for the identification of multiple targets. Then, we propose a \MTDiv method using dynamic programming techniques and improve its efficiency by leveraging the techniques of non-diverse selections and bounded pruning.

\subsection{Multiple Targets Scheme}

We start by presenting \emph{probability setting} and \emph{updating rules} for the \kw{Multiple}\kw{Targets} scheme.

\stitle{Target probability.} As there exist multiple targets with $|\tar|\geq 1$, we assume that each vertex has an independent probability of being a target. As our problem aims at finding $k$ selections and $|\tar|$ is unknown, let $\pr(v)$ be the vertex probability of $v$ and $\pr(v) = \frac{k}{n}$. Thus, the target probability of a vertex $v$ is computed as follows. The probability of $\reach(v)=\no$ is denoted as the \no probability $\pno(v) = \prod_{u \in \des(v) \cap \can}{(1 - \pr(u))}$, representing that none of vertices $u\in \des(v) \cap \can$ is a target. Moreover,  we  update the target probabilities with more questions asked. 

\stitle{Rules of updating $\can$ and $\yset$.} Following Def.~\ref{def.yset} and Lemmas~\ref{lemma.noprune} and \ref{lemma.yesprune}, we have the following rules for updating the potential targets $\can$ and \yescs $\yset$.

\begin{Rule}\label{lemma.MT.Yes}
For a vertex $v\in \can$ with $\reach(v) = \yes$, we update the potential targets $\can_{new} = \can \setminus (\anc(v)\setminus\{v\})$  and update the \yescs $\yset_{new}= \yset \cup \anc(v)$. 
\end{Rule}
\vspace{-0.2cm}
\begin{Rule}\label{lemma.MT.No}
For a vertex $v\in \can$ with $\reach(v) = \no$,  we update the potential targets $\can_{new} = \can \setminus \des(v)$ and keep the \yescs unchanged $\yset_{new}= \yset$. 
\end{Rule}

On the basis of Rules~\ref{lemma.MT.Yes} and \ref{lemma.MT.No}, we compute both $\gyes(v)$ and $\gno(v)$, which equals $\g(\yset, \can, k) - $ $\g(\yset_{new}, \can_{new}, k)$ by Def.~\ref{def.gain}.  However, as we know that $\g(\yset, \can, k) = \min_{\sset \subseteq \yset, |\sset| \leq k}$ ${\f(\sset, \can)}$, it is difficult to efficiently compute the penalty $\g(\yset, \can, k)$ with a straightforward enumeration of $\sset \subseteq \yset$ for $k> 1$. In the following sections, we mainly focus on developing efficient approaches to compute $\g(\yset, \can, k)$ and update $\yes$\&$\no$ gain scores $\gyes(v)$ and $\gno(v)$.

\subsection{kBM-DP Algorithm}

\begin{algorithm}[t]
  \small
  \caption{\MTDiv}
  \label{algo:multi}
  \begin{algorithmic}[1]
    \Require A hierarchy tree $\T=(V, E)$, a budget $b$, a number $k$.
      \Ensure Selections $\sset$ with $|\sset|\leq k$.
    \State Let  $\Syes \leftarrow \{r\}, \can \leftarrow V$;
    \State Assign the probability $\pr(v)=k/n$ for $v\in V$;
    \For {$i \leftarrow$ 1 to $b$}
      \For {$v \in \can \setminus \yset$}
        \State $\pno(v) \leftarrow \prod_{u \in \des(v) \cap \can}{(1 - \pr(u))}$; 
        \State $\pyes(v)  \leftarrow  1 - \pno(v)$;\
      \EndFor
      \State Calculate $\eg(v)$ for all $v\in \can \setminus \yset$ using Algorithm~\ref{algo:mtdiv};
      \State $q_i \leftarrow \arg\max_{v \in \can \setminus \yset} \eg(v)$;
      \State Ask the question $\reach(q_i)$;
      \If {$\reach(q_i) = \yes$}
        \State $\yset \leftarrow \yset \cup \anc(q_i)$ by Rule~\ref{lemma.MT.Yes};
        \State $\can_{new} \leftarrow \can \setminus (\anc(q_i) \setminus \{q_i\})$ by Rule~\ref{lemma.MT.Yes};
      \Else
        \State $\can_{new} \leftarrow \can \setminus \des(q_i)$ by Rule~\ref{lemma.MT.No};
      \EndIf
       \State Update $\pr(u) = 0$ for $u \in \can \setminus \can_{new}$;
       \State Update $\pr(u) = \pr(u) \cdot \frac{|\can|}{|\can_{new}|}$ for $u \in \can_{new}$;
      \State $\can \leftarrow \can_{new}$;

      \If {$\can \subseteq \yset$} \textbf{break};
      \EndIf
    \EndFor
    \State $\sset^* \leftarrow \arg \min_{\sset \subseteq \yset, |\sset| \leq k}{\f(\sset, \can)}$;
    \State \Return{$\sset^*$};
  \end{algorithmic}
\end{algorithm}

\begin{algorithm}[t]
  \small
  \caption{\MTDiv: Calculate Expected Gains} \label{algo:mtdiv}
  \begin{algorithmic}[1]
    \Require $\T=(V, E)$, $\pyes(.)$, $\pno(.)$, $\can$, $\yset$, and $k$.
    \Ensure $\eg(v)$ for all vertices $u \in \can \setminus \yset$.

    \State Calculate $\DP(u, w, k)$ for all vertices $u \in \can$ and $ w \in \anc(u)$ in Eq.~\ref{eq.dp}.
    \For {$v \in \can \setminus \yset$} 
      \State $\cyes \leftarrow \can \setminus (\anc(v) \setminus \{v\})$; \hspace{1 cm}// $\reach(v)=\yes$
      \State $\gyes(v) = \g(\yset, \can, k) - \calyes(v, k)$;
      \State $\cno \leftarrow \can \setminus \des(v)$;            \hspace{2 cm}// $\reach(v)=\no$
      \State $\gno(v) = \g(\yset, \can, k)  - \calno(v, k)$;
      \State $\eg(v) \leftarrow \gyes(v) \cdot \pyes(v) + \gno(v) \cdot \pno(v)$ by Def.~\ref{def.eg};
    \EndFor
    \Procedure {$\calyes$}{$u, k$}
    \For {$v \in \anc(u)$}
      \State Recalculate $\DP_{Y}(v, k)$ by Eq.~\ref{eq.dp.yes};
      \For {$w \in \anc(v) \setminus \{v\}$}
        \State Recalculate $\DP_{N}(v, w, k)$ by Eq.~\ref{eq.dp.no};
        \State $\DP(v, w, k) \leftarrow \min\{\DP_{N}(v, w, k), \DP_{Y}(v, k)\}$;
      \EndFor
    \EndFor
    \State $\g(\yyes, \cyes, k) \leftarrow \DP(r, r, k)$;
    \State \Return{$\g(\yyes, \cyes, k)$};
    \EndProcedure
    \Procedure {$\calno$}{$u, k$}
    \For {$v \in \anc(u)$}
      \For {$w \in \anc(v) \cap \yset \setminus \{v\}$}
        \State Recalculate $\DP_{N}(v, w, k)$ by Eq.~\ref{eq.dp.no};
      \EndFor
      \If{$v \in \yset$}
        \State Recalculate $\DP_{Y}(v, k)$ by Eq.~\ref{eq.dp.yes};
        \State $\DP(v, w, k) \leftarrow \min\{\DP_{N}(v, w, k), \DP_{Y}(v, k)\}$;
      \Else
        \State $\DP(v, w, k) \leftarrow \DP_{N}(v, w, k)$;
      \EndIf
    \EndFor
    \State $\g(\yno, \cno, k) \leftarrow \DP(r, r, k)$;
    \State \Return{$\g(\yno, \cno, k)$};
    \EndProcedure
  \end{algorithmic}
\end{algorithm}

In this section, we propose a \MTDiv algorithm for identifying multiple targets based on the \Frame framework in Algorithm~\ref{algo:framework}. 

\stitle{The kBM-DP algorithm}. The algorithm of \MTDiv is presented in Algorithm~\ref{algo:multi}, which implements the details of Algorithm~\ref{algo:framework}. The algorithm first initializes the \yescs $\yset$, the potential targets $\can$, and the independent probability  $\pr(v)=\frac{k}{n}$ for each vertex $v\in V$ (lines 1-2). Then, it iteratively selects one best vertex $v\in \can \setminus \yset$ with the largest $\eg(v)$ and asks question $\reach(v)$ until all $b$ questions have been asked (lines 3-18). At the $i$-th round of asking question, it updates the target probabilities for all vertices $ \can \setminus \yset$ (lines 4-6). The algorithm invokes Algorithm~\ref{algo:mtdiv} to calculate all expected gains (line 7).  Next, the algorithm asks question $\reach(q_i)$, and updates the vertex probabilities, $\can$ and $\yset$ accordingly by Rules~\ref{lemma.MT.Yes} and ~\ref{lemma.MT.No} (lines 10-16). Finally, the algorithm returns the selections  $\sset^* = \arg \min_{\sset \subseteq \yset, |\sset| \leq k}{\f(\sset, \can)}$ (line 19). 

In the following, we introduce a dynamic programming algorithm for calculating the expected gains in Algorithm~\ref{algo:mtdiv}. We first consider the computation of $\g(\yset, \can, k)$ and then $\yes$\&$\no$ gain scores.

\begin{figure}[t]
\vspace{-0.2cm}
\centering
{
\subfigure{
\includegraphics[width=0.4\linewidth]{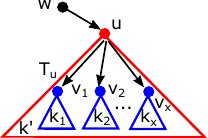} }
}
\vspace{-0.4cm}
\caption{A solution overview of dynamic programming.}
\label{fig.dp.overview}
\vspace{-0.8cm}
\end{figure}

\stitle{Computing $\g(\yset, \can, k)$}. 
An intuitive approach enumerates all $k$-sized selections $\sset\subseteq \yset$ for finding the best selections $S^*$, which is inefficient. Thus, we use a dynamic programming technique to calculate the minimum $\g(\yset, \can, k)$ efficiently. The general idea is to divide the global calculation into sub-problems of finding $k^\prime \leq k$ selection vertices optimally in a subtree $T_u$ rooted by a vertex $u$. The vertex $w \in\sset \cap \anc(u)$ is a selection closest to $u$. 
Note that if $\sset \cap \anc(u) = \emptyset$, we consider $w = r$. 
Obviously, let $u = r, w = r$ and $k^\prime = k$, this subproblem is the same as the best selection of $S$. Thus, our objective is to calculate it from the sub-problems. We consider two cases of whether we select vertex $u$ or not for each subtree $T_u$. On one hand, if $u \in \yset$ and we select vertex $u$ into the selections $\sset$, for each child node $v_1, v_2 ..., v_x \in \child(u) \cap \can$, the sub-problem is how to find additional $k_x$ optimal vertices in the subtrees rooted by $v_x$ with the closest selected vertex $u$ and $\sum{k_x} \leq k^\prime - 1$; On the other hand, if we do not select vertex $u$ into the answer $\sset$, for each child node $v_1, v_2 ..., v_x \in \child(u) \cap \can$, the sub-problem is how to find additional $k_x$ optimal vertices in $T_x$ with the closest selected vertex $w$ and $\sum{k_x} \leq k^\prime$. The optimal answer is the best solution among the above two answers.

\stitle{State and transfer equation.} We first define the state of $\DP_{Y}(u, k)$ and $\DP_{N}(u, w, k)$. In the subtree $T_u$, $\DP_{Y}(u, k)$ is the minimum value of $\f(\sset \cup \{u\}, \can \cap \des(u))$ with selected $k - 1$-size set $\sset \subseteq \des(u)$. Similarly, $\DP_{N}(u, w, k)$ is the minimum value of $\f(\sset \cup \{w\}, \can \cap \des(u))$ with selected $k$-size set $\sset \subseteq \des(u) \setminus \{u\}$ and $w \in \anc(u) \cap \yset$ is the closest selected vertex to $u$. On the basis of $\DP_{Y}(u, k)$ and $\DP_{N}(u, w, k)$, we define the state $\DP(u, w, k)$ as the optimal $k$-size selection in the subtree $T_u$ with closest selected vertex $w \in \anc(u) \cup \can$, which satisfies the equation as follows,
\begin{equation} \label{eq.dp}
\DP(u, w, k) = 
\left\{
\begin{aligned}
&\min\{\DP_{Y}(u, k), \DP_{N}(u, w, k)\}, \hspace{-0.6cm}  & u \in \yset, k \geq 1\\
&\DP_{N}(u, w, k), &u \notin \yset\ or\ k = 0 \\
\end{aligned}
\right.
\end{equation}

Next, we propose the transfer equation of $\DP_{Y}(u, k)$ and $\DP_{N}(u,$ $ w, k)$ as follows, 
\vspace{-0.1cm}
\begin{equation}\label{eq.dp.yes}
\begin{aligned}
\DP_{Y}(u, k) & = \min\{\sum_{x \in \child(u) \cap \can} \DP(x, u, k_x)\}\\ 
&\text{ subject to }  \sum_{x\in \child(u) \cap \can} k_x= k-1.\\
\end{aligned}
\end{equation}
\vspace{-0.3cm}
\begin{equation}\label{eq.dp.no}
\begin{aligned}
\DP_{N}(u, w, k) & = \dist \langle w, u \rangle + \min\{\sum_{x \in \child(u) \cap \can} \DP(x, w, k_x)\}\\ 
&\text{ subject to }  \sum_{x\in \child(u) \cap \can} k_x= k.\\
\end{aligned}
\end{equation}
\vspace{-0.3cm}

Furthermore, we can use the Knapsack dynamic programming technique~\cite{pisinger1995algorithms} to tackle the transfer Eqs.~\ref{eq.dp.yes} and \ref{eq.dp.no}. Assume that a number $k$ represents the total capacity. Given a set of vertices $\child(u) \cap \can = \{x_1, ..., x_l\}$, for each vertex $x_i$ where $1\leq i\leq l$, $\DP(x_i, w, k_{x_i})$ represents an item value and the item volume is $k_{x_i}\leq k$. We assume that $F(i, k')$ is the state that has the minimum value of the first $i$ items with a total of $k'$ capacity. The equation of state transformation is shown as follows.
$$F(i, k') = \min_{0 \leq j \leq k'}\{F(i - 1, k' - j) + \DP(x_i, w, j)\}.$$
For initialization, we set $F(i, j) = +\infty$ for $1\leq i\leq l, 0\leq j\leq k$ and $F(0, 0) = 0$. The return value is $F(l, k)= \min \{\sum_{x \in \child(u) \cap \can}$ $\DP(x, w, k_x)\}$ with the constraint $\sum_{x\in \child(u) \cap \can} k_x= k$.

\begin{table}[t]
\vspace{-0.2cm}
\caption{The value of $\gyes$, $\gno$, and $\eg$. Here, $b = 2$, $k = 2$ and $\{v_5, v_8\}$ are the targets.}
\vspace{-0.4cm}
\centering
\scalebox{0.75}{
\begin{tabular}{c|ccccccccc}
\toprule
Node& $v_{1}$& $v_{2}$& $v_{3}$& $v_{4}$& $v_{5}$ &$v_{6}$ &$v_{7}$& $v_8$& $v_9$\\
\midrule
$\gyes$& 8& 1& 12& 9& 10& 12& 12& 12& 11\\
$\gno$& 19& 1& 11& 2& 5& 3& 3& 3& 3\\
$\eg^1$& 9.85& 1& \textbf{11.59}& 3.4& 6.8& 4.8& 4.8& 4.8& 4.6\\
\hline
$\gyes$& /& 0& /& 0& 1& 0& 0& 0& 2\\
$\gno$& /& 1& /& 1& 3& 1& 1& 1& 2\\
$\eg^2$& /& 0.75& /& 0.75& \textbf{2.12}& 0.75& 0.75& 0.75& 2\\
\bottomrule
\end{tabular}
}
\label{table:g2}
\vspace{-0.5cm}
\end{table}

\stitle{Update $\gyes(u)$ and $\gno(u)$.} For a vertex $u$ that is questioned, 
only the state of $\anc(u)$ needs to be recalculated to update $\gyes(u)$ and $\gno(u)$. Algorithm~\ref{algo:mtdiv} presents the details of $\gyes(u)$ and $\gno(u)$ calculation. First, the algorithm calculates $\DP(u, w, k)$ for all possible states (line 1). Then, for each vertex $u \in \can \setminus \yset$, it updates the corresponding states and calculates $\gyes(u)$ and $\gno(u)$ (lines 2-7). For computing $\gyes(u)$, it needs to update the states $\DP(v, w, k)$ for all $v \in \anc(u)$, $w \in \anc(v)$ (lines 8-15). To computing $\gno(u)$, it updates the states $\DP(v, w, k)$ for all $v \in \anc(u)$, $w \in \anc(v) \cap \yset$ (lines 16-26).

\begin{example}
Assume the targets are $\tar = \{v_5, v_8\}$, the budget $b = 2$, and $k = 2$ in Figure~\ref{example.prob}. Table~\ref{table:g2} shows the gains of each vertex. \revision{In the first question, the algorithm questions $v_3$ and gets the $\yes$ answer.} Then, the vertices $v_0, v_1$ will be pruned from the potential targets. The vertex $v_5$ gets the maximum gains in the second round and obtains the $\yes$ answer. Note that some leaf vertices get $\gyes(v) = 0$ because their parents are a better selection even if they get the $\yes$ answer. The selections will be $\sset = \{v_3, v_5\}$ and the penalty will be $\f(\sset, \tar) = 1$.
\end{example}

\stitle{Complexity analysis.} The calculation of all states takes $O(nhk^2)$ time. The update of each vertex takes $O(h^2dk^2)$ times using Algorithm~\ref{algo:mtdiv}. Overall, the \MTDiv in Algorithm~\ref{algo:multi} takes $O(bnh^2dk^2)$ time in $O(nhk)$ space for generating $b$ questions.


\subsection{Fast algorithms: kBM-Topk and kBM-DP+}

In this section, we propose two fast algorithms of \MTTopk and \MTfast. The first method \MTTopk uses an alternative penalty function to improve the calculation of expected gain. The second method \MTfast develops an upper bound of $\eg(v)$ to prune unnecessary vertices for updating the expected gains, which achieves the same results as \MTDiv in most cases.

\subsubsection{kBM-Topk}
The penalty of $\g(\yset, \can, k)$ is complex to compute, due to the dependence relationship of selections in $\sset$. To deal with this issue, we propose an independent function to approximate $\g(\yset, \can, k)$, which can be efficiently computed. We begin with a new definition of selected gain as follows. 
$$\rdc(x) = \sum_{v \in \can \cap \des(x)}{dist \langle r, x \rangle}.$$



The selected gain represents the reduced penalty after selecting $x$. The deeper the selected vertex and the larger the size of descendants, the higher the selected gain. Thus, the general idea of the approximate method is to select the top-$k$ vertices that maximize the selected gain. On this basis, we propose the new potential penalty function as follows,
$$\g^\prime(\yset, \can, k) = \f(\{r\}, \can) - \max_{\sset \subseteq \yset, |\sset| \leq k}{\sum_{x \in \sset}{\rdc(x)}}.$$


\stitle{kBM-Topk algorithm.} Similar to \MTDiv, \MTTopk uses the same framework and the difference is the calculation of expected gains. Algorithm~\ref{algo:mttopk} shows the details of the calculation (replacing the line 7 of Algorithm~\ref{algo:multi}). First, the algorithm calculates $\red(v)$ for all the vertices $v \in \can \cup \yset$ and uses max-heap $H$ to maintain the top-$k$ vertices in \yescs $\yset$ (lines 1-2). Then, for each vertex $u \in \can \setminus \yset$, it updates the max-heap $H$ and calculates $\gyes(u)$ and $\gno(u)$ (lines 3-8). For the $\gyes(u)$ calculation, $\red(v)$ of vertex $v \in \anc(u)$ should be updated (lines 10-11). If the vertex $v$ has been added into the heap $H$, $\red(v)$ will be replaced by $\red_{new}(v)$ (lines 12-13). Otherwise, it will be added into the heap $H$ because $v \in \yset_{new}$ (lines 14-15). Finally, the algorithm updates the $\f(\{r\}, \can)$ and calculates $\gyes(u)$ (lines 16-17). The calculation of $\gno(u)$ is similar to that of $\gyes(u)$ (lines 18-23). The main difference is that the \yescs $\yset$ will not be changed and only the vertices $v \in \anc(u) \cap \yset$ should be updated (lines 19-21).

\begin{algorithm}[t]
  \small
  \caption{\MTTopk: Calculate Expected Gains}
  \label{algo:mttopk}
  \begin{algorithmic}[1]
    \Require $\T=(V, E)$, $\pyes(.)$, $\pno(.)$, $\can$, $\yset$, and $k$.
    \Ensure $\eg(v)$ for all vertices $u \in \can \setminus \yset$.
    
    \State Calculate $\red(u)$ for all vertices $u \in \can \cup \yset$.
    \State Let $H$ be a max-heap to maintain $\red(u)$ for all vertices $u \in \yset$.
    \For {$v \in \can \setminus \yset$} 
      \State $\cyes \leftarrow \can \setminus (\anc(v) \setminus \{v\})$; \hspace{1 cm}// $\reach(v)=\yes$
      \State $\gyes(v) = \g(\yset, \can, k) - \calyes(v, k, H)$;
      \State $\cno \leftarrow \can \setminus \des(v)$;            \hspace{2 cm}// $\reach(v)=\no$
      \State $\gno(v) = \g(\yset, \can, k)  - \calno(v, k, H)$;
      \State $\eg(v) \leftarrow \gyes(v) \cdot \pyes(v) + \gno(v) \cdot \pno(v)$ by Def.~\ref{def.eg};
    \EndFor
    \Procedure {$\calyes$}{$u, k, H$}
    \For {$v \in \anc(u)$}
      \State $\red_{new}(v) = \red(v) - |\can \setminus \cyes \cap \des(v)| \cdot \dist \langle r, v \rangle$;
      \If{$v \in Y$}
        \State $H(v, \red(v)) \leftarrow (v, \red_{new}(v))$;
      \Else
        \State $H.push(v, \red_{new}(v))$;
      \EndIf
    \EndFor
    \State $\f(\{r\}, \cyes) = \f(\{r\}, \can) - \sum_{v \in \can \setminus \cyes}{\dist \langle r, v \rangle}$;
    \State \Return{$\f(\{r\}, \cyes) - \max_{\sset \subseteq H, |\sset| \leq k}{\sum_{x \in S}{\red_{new}(x)}}$};
    \EndProcedure
    \Procedure {$\calno$}{$u, k, H$}
    \For {$v \in \anc(u) \cap Y$}
      \State $\red_{new}(v) = \red(v) - |\can \setminus \cno| \cdot \dist \langle r, v \rangle$;
      \State $H(v, \rdc(\{v\}, \can)) \leftarrow (v, \rdc_{new}(\{v\}, \can))$;
    \EndFor
    \State $\f(\{r\}, \cno) = \f(\{r\}, \can) - \sum_{v \in \can \setminus \cno}{\dist \langle r, v \rangle}$;
    \State \Return{$\f(\{r\}, \cno) - \max_{\sset \subseteq H, |\sset| \leq k}{\sum_{x \in \sset}{\red_{new}(x)}}$};
    \EndProcedure
  \end{algorithmic}
\end{algorithm}



\stitle{\revision{Quality and complexity analysis.}} \revision{We analyze the quality of our approximate penalty function $\g^\prime(\yset, \can, k)$ as follows.}

\begin{theorem}
\label{theorem.topk}
\revision{
It holds that 
$\kw{LB}_{k}\leq \g^\prime(\yset, \can, k) \leq\kw{UB}_{k}$, where $\kw{LB}_{k}= \g(\yset, \can, k) - (k - 1) \cdot (\f(\{r\}, \can) - \g(\yset, \can, k)) $ and $\kw{UB}_{k}= \g(\yset, \can, k) $. 
If $k=1$,  $\g^\prime(\yset, \can, 1) = \g(\yset, \can, 1)$. 
}
\end{theorem} 

\begin{proof}
\revision{First, we prove $\g^\prime(\yset, \can, k) \leq \kw{UB}_{k}$. Assume that $\sset^*$ is the best selection in terms of potential penalty, i.e., $\g(\yset, \can, k) = \f(\sset^*, \can)$. Then $\g(\yset, \can, k) = \f(\sset^*, \can) \geq \f(\{r\}, \can) - \sum_{x \in \sset^*}{\rdc(x)} \geq \g^\prime(\yset, \can, k).$ Thus, $\g^\prime(\yset, \can, k) \leq \g(\yset, \can, k)=\kw{UB}_{k}$. } \revision{Second, we prove $\g^\prime(\yset, \can, k) \geq \kw{LB}_{k}$. We have $\g^\prime(\yset, \can, k) \geq \f(\{r\}, \can) - k \cdot \max_{x \in \yset}{\rdc(x)} = k \cdot (\f(\{r\}, \can) - \max_{x \in \yset}{\rdc(x)}) - (k - 1) \cdot \f(\{r\}, \can) \geq k \cdot \g(\yset, \can, k) - (k - 1) \cdot \f(\{r\}, \can) = \g(\yset, \can, k) - (k - 1) \cdot (\f(\{r\}, \can) - \g(\yset, \can, k))$$=\kw{LB}_{k}$. }
\revision{Overall, $\kw{LB}_{k}\leq \g^\prime(\yset, \can, k) \leq\kw{UB}_{k}$. For $k=1$, we have $\kw{UB}_{k}=\g(\yset, \can, k) =\kw{LB}_{k}$. Thus,  $\g^\prime(\yset, \can, 1) = \g(\yset, \can, 1)$ holds.}
\end{proof}



\revision{Next, we analyze the complexity of  \MTTopk.} For each vertex $u$ in $\can \setminus \yset$, at most $|\anc(u)|$ vertices are calculated and each update in heap $H$ takes $O(\log{|\yset|})$ time. Thus, the score calculation of each vertex $u$ takes $O(h \log{n})$. Overall, \MTTopk takes $O(bnh\log{n})$ time using $O(n)$ space for generating $b$ questions.

\subsubsection{kBM-DP+} 
\label{sec.dp+}
\begin{algorithm}[t]
  \small
  \caption{\MTfast: Calculate Expected Gains and Identify $q_i$} \label{algo:mtdiv+}
  \begin{algorithmic}[1]
    \Require $\T=(V, E)$, $\pyes(.)$, $\pno(.)$, $\can$, $\yset$, and $k$.
    \Ensure Question vertex $q_i$.
    \State $\eg_{max} \leftarrow 0$;
    \State $\overline{\eg^{i}}(v) \leftarrow \UB\gyes^i(v) \cdot \pyes(v) + \UB\gno^i(v) \cdot \pno(v)$ for all vertices $v \in \can \setminus \yset$, where $\UB\gyes^i(v) = \gyes^{(i-1)}(v)$ and $\UB\gno^i(v) = \gno^{(i-1)}(v)$.
    \State Sort all vertices $v\in \can \setminus \yset$ in the descending order of $\overline{\eg^{i}}(v)$;
    \For {$v \in \can \setminus \yset$} 
      \If{$\eg_{max} > \overline{\eg^{i}}(v)$} \Return{$q_i$};
      \EndIf
      \State $\gyes(v) = \g(\yset, \can, k) - \calyes(v, k)$;
      \State $\gno(v) = \g(\yset, \can, k)  - \calno(v, k)$;
      \State $\eg(v) \leftarrow \gyes(v) \cdot \pyes(v) + \gno(v) \cdot \pno(v)$;
      \If{$\eg_{max} < \eg(v)$} 
        \State $\eg_{max} \leftarrow \eg(v)$; $q_i \leftarrow v$;
      \EndIf
    \EndFor
    \State \Return{$q_i$};
  \end{algorithmic}
\end{algorithm}

In this section, we propose a pruning optimization to accelerate the algorithm \MTDiv. 
The general idea is to design an upper bound of expected gain and skip the update of $\eg(v)$ for those vertices that are disqualified for achieving the largest  $\eg(v)$ in $\can \setminus \yset$. In this way, we can prune lots of vertices in most cases at each round of question asking, and quickly identify a vertex $q_i$ with the largest gain.  

\begin{table*}[t!]
\vspace{-0.3cm}
\caption{\revision{The averaged penalty of HGS, IGS, BinG and STBIS in SingleTarget problem on all datasets varied by budget $b$.}}
\vspace{-0.4cm}
\centering
\scalebox{0.9}{
\begin{tabular}{|c|cccc|cccc|cccc|cccc|}
\toprule
$\ $& \multicolumn{4}{|c|}{\revision{Image-COCO}}& \multicolumn{4}{|c|}{ImageNet}& \multicolumn{4}{|c|}{Yago3-I}& \multicolumn{4}{|c|}{Yago3-II} \\
\midrule
Budget $b$
& HGS& IGS& BinG& STBIS& HGS& IGS& BinG& STBIS& HGS& IGS& BinG& STBIS& HGS& IGS& BinG& STBIS\\
\hline
5& 1.48& 1.42& 1.21& \textbf{1.20}& 4.59& 4.50& 3.98& \textbf{3.89}& 2.04& 2.03& \textbf{1.60}& 1.61& 2.61& 2.60& 2.30& \textbf{2.27}\\
10& 1.34& 1.15& 0.79& \textbf{0.78}& 4.50& 3.87& 3.10& \textbf{3.04}& 1.92& 1.63& 1.15& \textbf{1.05}& 2.48& 2.20& 1.69& \textbf{1.53}\\
20& 1.20& 0.77& \textbf{0.47}& \textbf{0.47}& 4.34& 3.02& 1.80& \textbf{1.79}& 1.83& 1.18& 0.71& \textbf{0.65}& 2.30& 1.61& 1.12& \textbf{1.05}\\
50& 0.98& 0.41& \textbf{0.18}& \textbf{0.18}& 4.21& 1.41& \textbf{0.67}& \textbf{0.67}& 1.71& 0.82& 0.41& \textbf{0.39}& 2.14& 1.04& 0.73& \textbf{0.67}\\
100& 0.52& 0.19& \textbf{0.00}& \textbf{0.00}& 3.79& 0.68& \textbf{0.28}& \textbf{0.28}& 1.66& 0.71& 0.30& \textbf{0.27}& 2.07& 0.90& 0.56& \textbf{0.54}\\
\bottomrule
\end{tabular}
}
\label{table.exp1}
\end{table*}

\begin{figure*}[t!]
\vspace{-0.4cm}
\centering
{ 
\subfigure[\revision{Image-COCO}]{
\label{fig.exp2_1}
\includegraphics[width=0.225\linewidth]{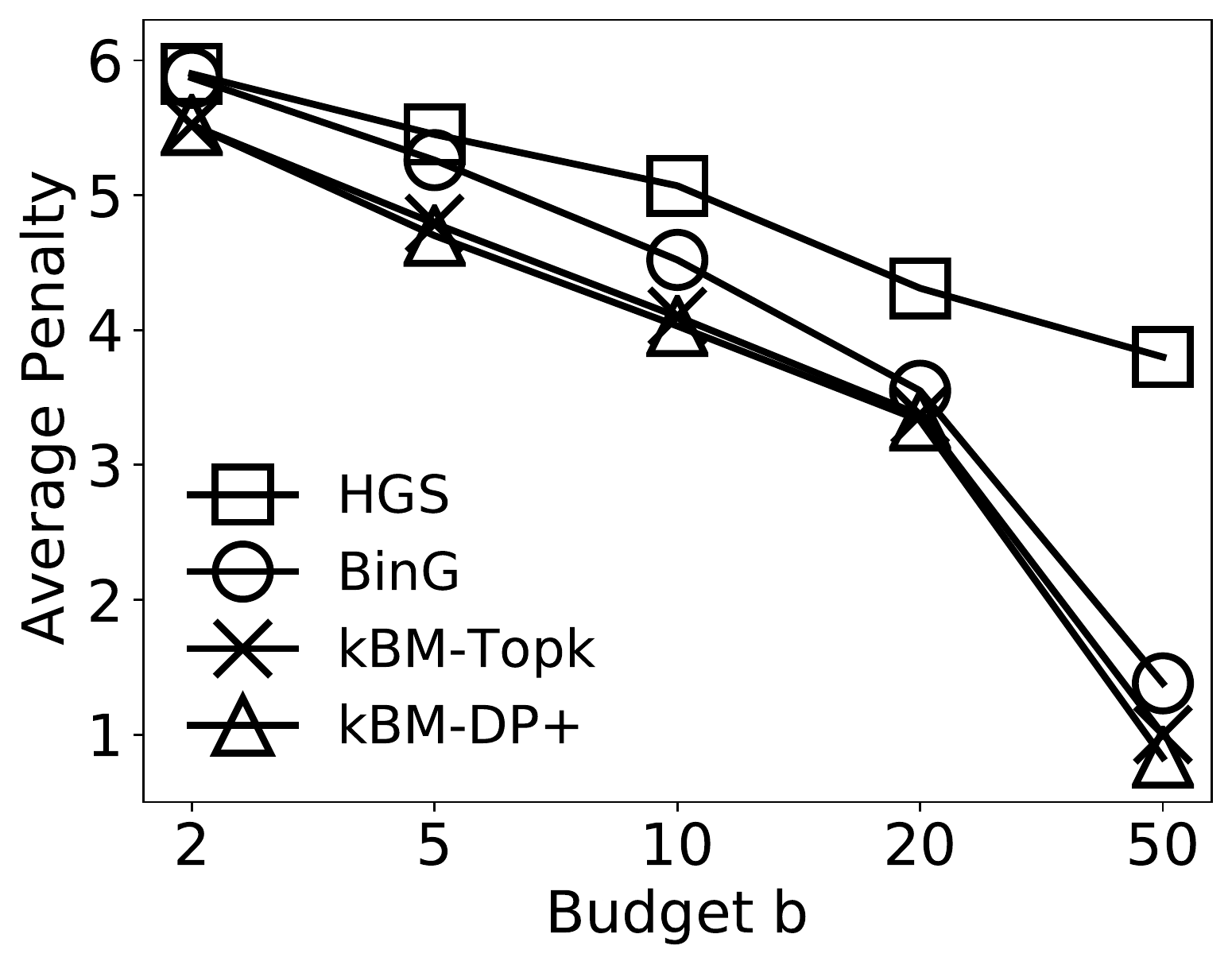} }
\subfigure[ImageNet]{
\label{fig.exp2_2}
\includegraphics[width=0.235\linewidth]{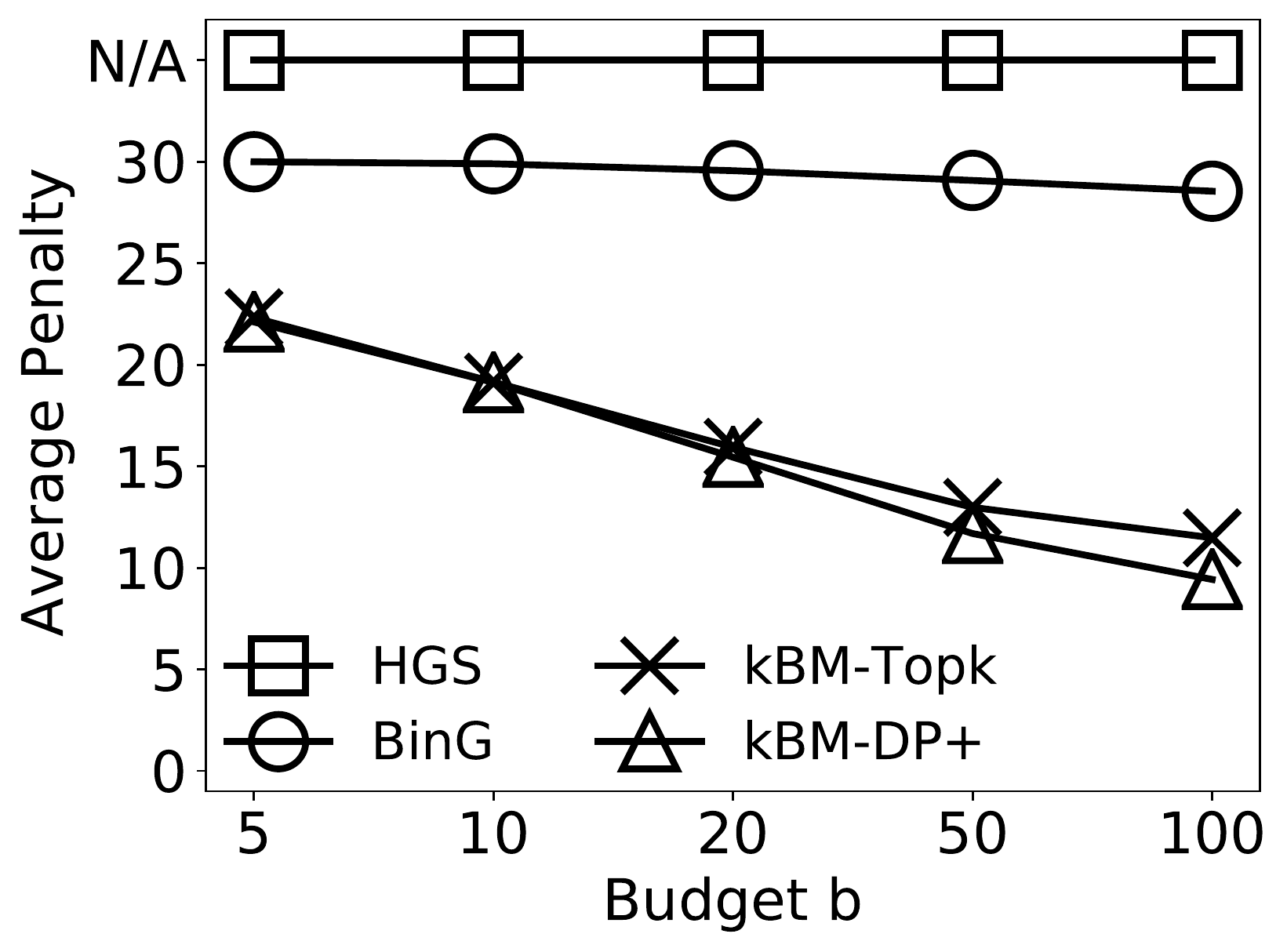} }
\subfigure[Yago3-I]{
\label{fig.exp2_3}
\includegraphics[width=0.235\linewidth]{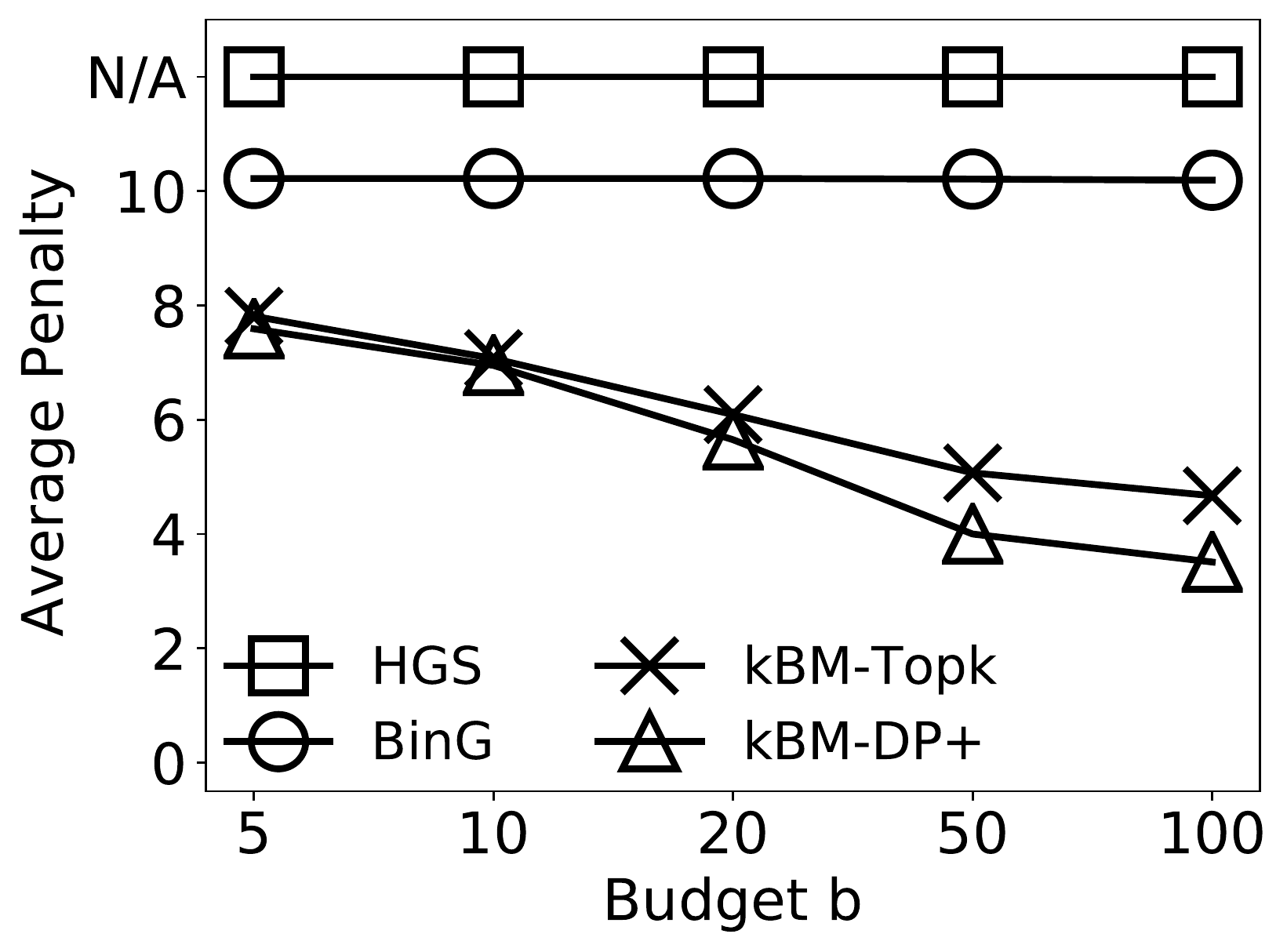} }
\subfigure[Yago3-II]{
\label{fig.exp2_4}
\includegraphics[width=0.235\linewidth]{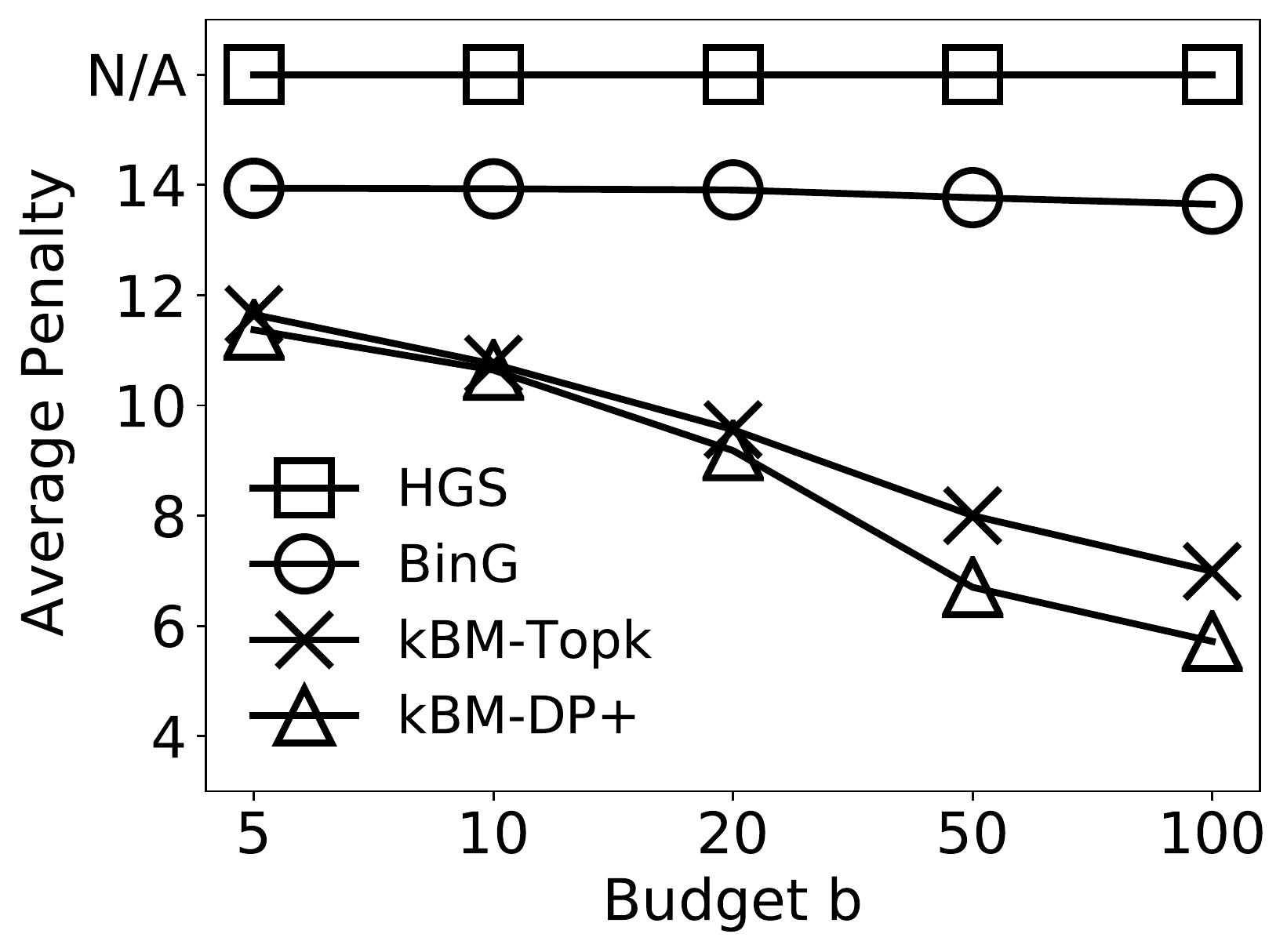} }
\vspace{-0.3cm}
\quad
\subfigure[\revision{Image-COCO}]{
\label{fig.exp2_5}
\includegraphics[width=0.225\linewidth]{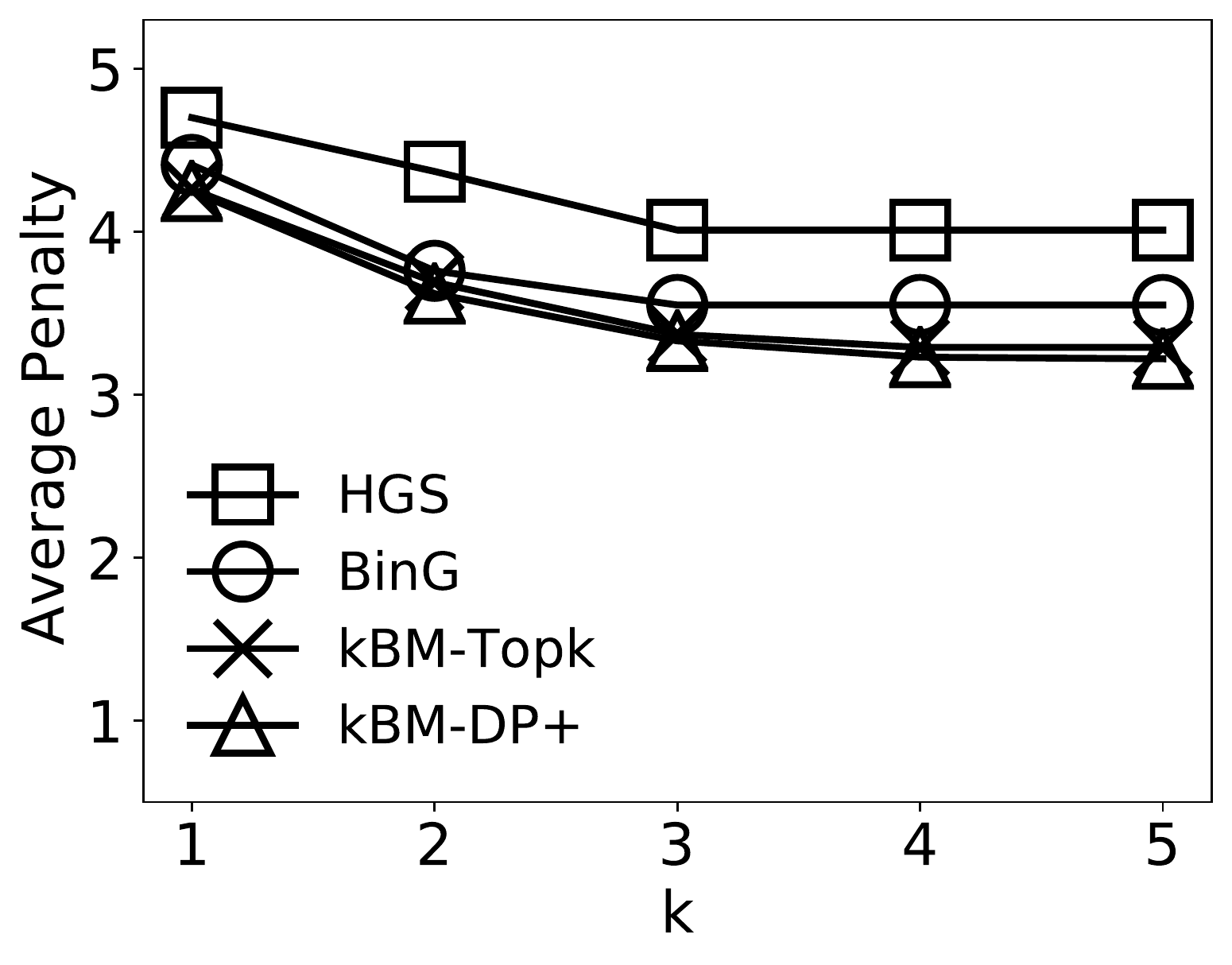} }
\subfigure[ImageNet]{
\label{fig.exp2_6}
\includegraphics[width=0.235\linewidth]{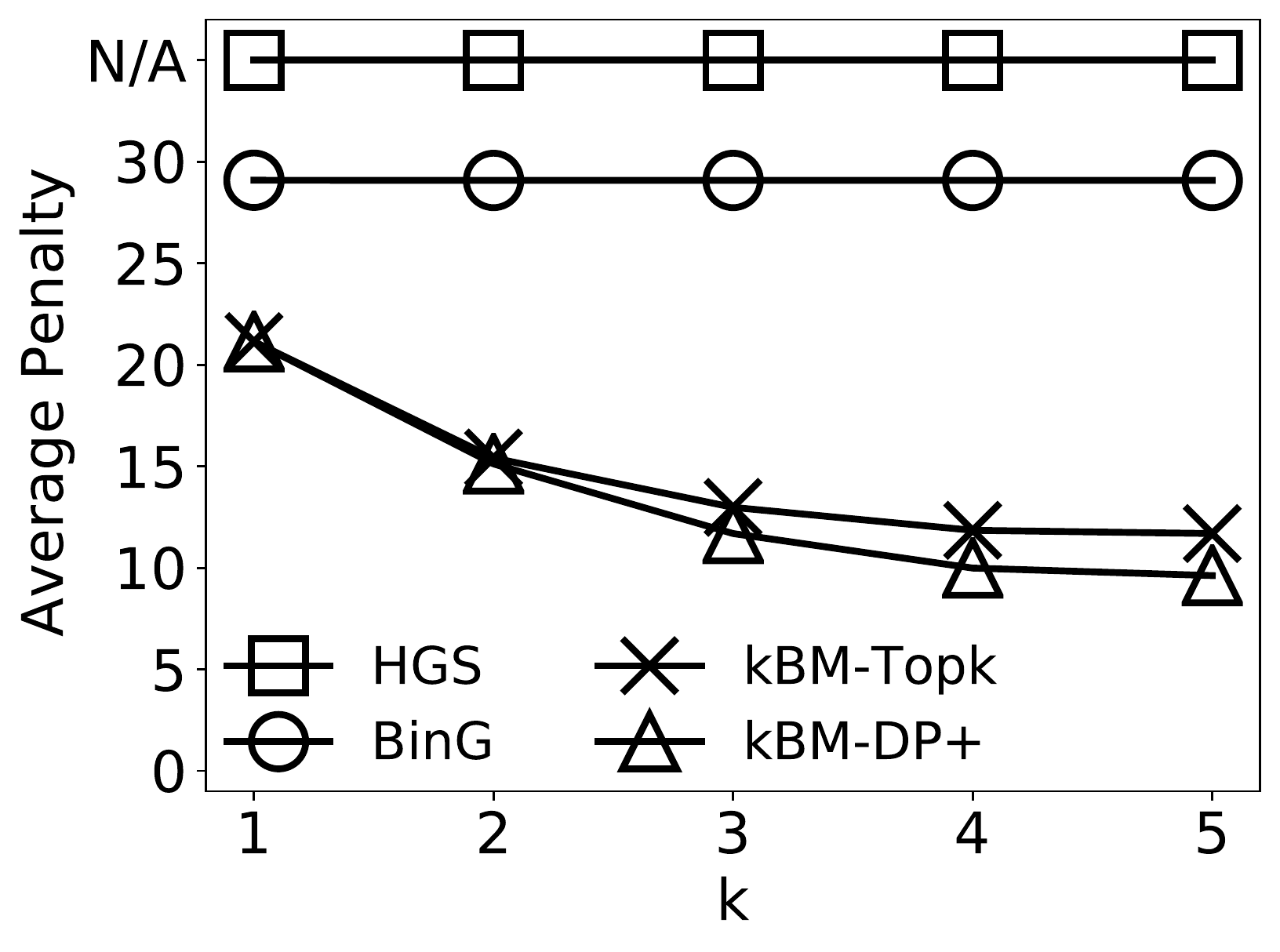} }
\subfigure[Yago3-I]{
\label{fig.exp2_7}
\includegraphics[width=0.235\linewidth]{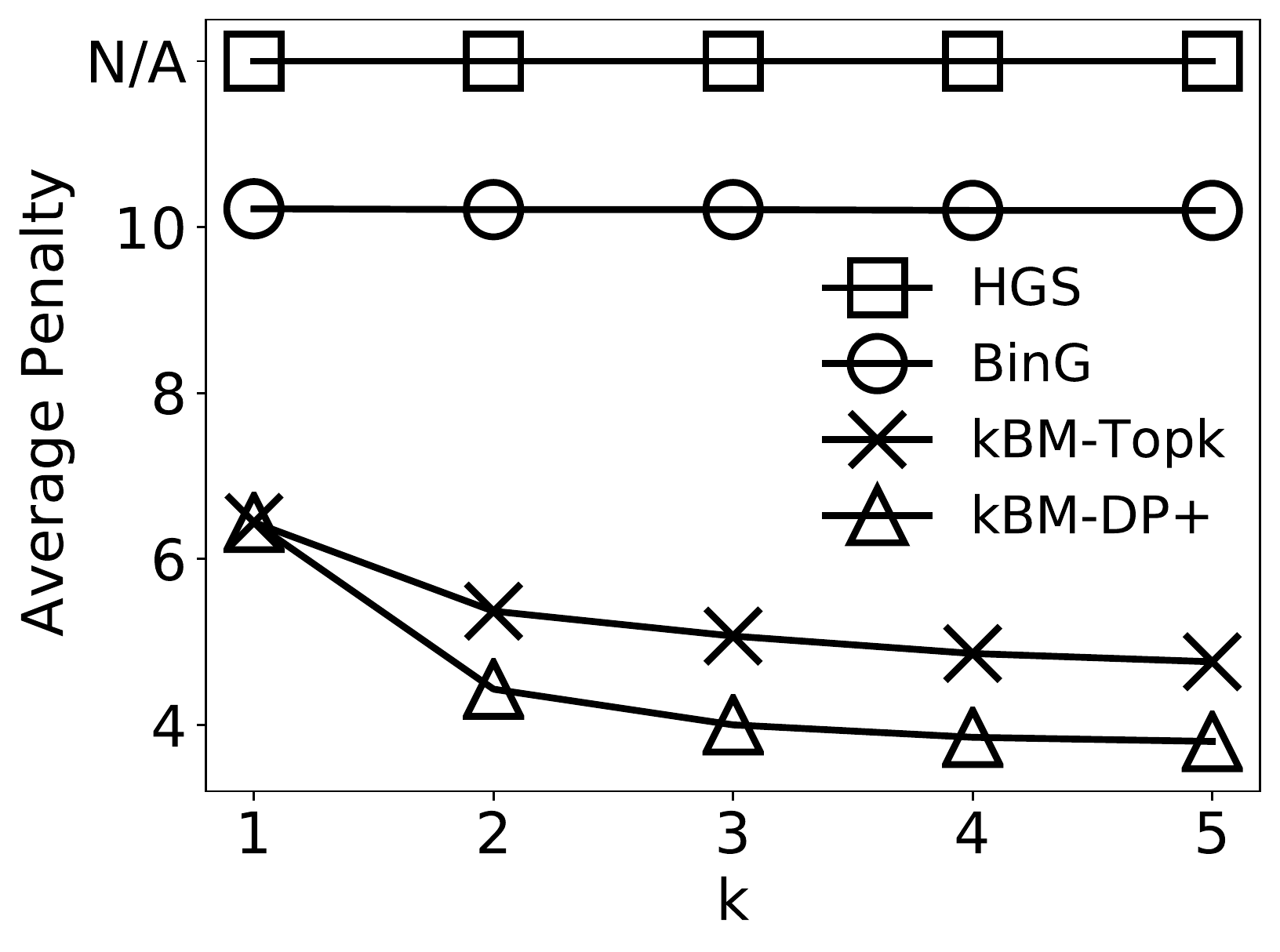} }
\subfigure[Yago3-II]{
\label{fig.exp2_8}
\includegraphics[width=0.235\linewidth]{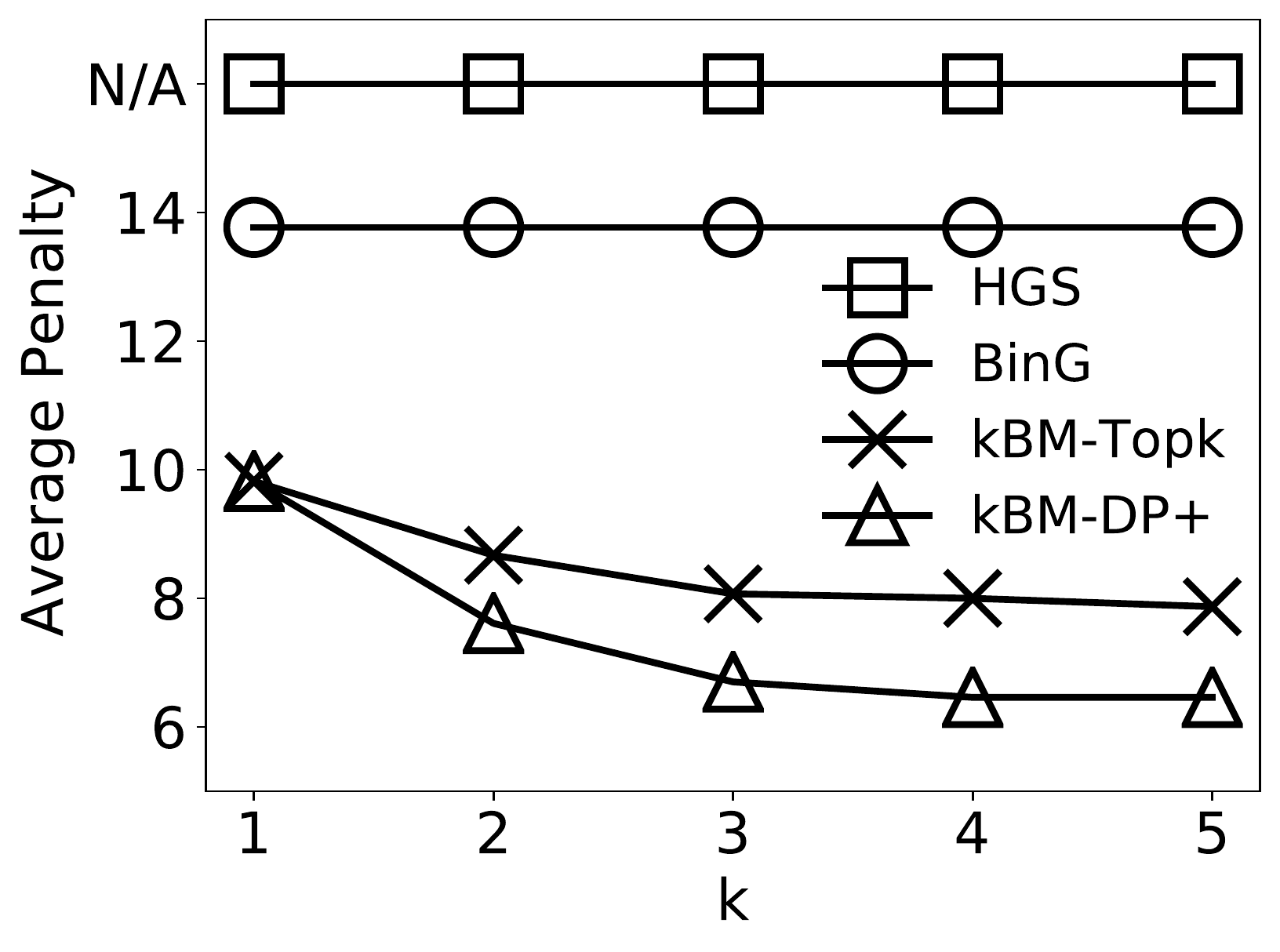} }
}
\vspace{-0.5cm}
\caption{Quality evaluation of all methods for MultipleTargets problem, in terms of the averaged penalty. 
}
\label{fig.exp2}
\vspace{-0.3cm}
\end{figure*}

\stitle{An upper bound of $\eg(v)$}. Consider a vertex $v\in V$ at the $i$-th round of question asking where $i \geq 1$, the expected gain is denoted as $\eg^{i}(v)$. Then, we have an upper bound of $\eg^{i}(v)$, denoted as $\overline{\eg^{i}}(v)$, satisfying 
$$\overline{\eg^{i}}(v)= \UB\gyes^i(v) \cdot \pyes(v) + \UB\gno^i(v) \cdot \pno(v),$$
where $\UB\gyes^i(v) = \gyes^{(i-1)}(v)$ and $\UB\gno^i(v) = \gno^{(i-1)}(v)$. Note that if $v$ is pruned in the previous round, we will set $\UB\gyes^i(v)$ $ = \UB\gyes^{(i-1)}(v)$ and $\UB\gno^i(v) = \UB\gno^{(i-1)}(v)$. We observe that both the $\yes$ gain and $\no$ gain decrease with more questions asked in most cases. This is because that the potential targets decrease and the \yescs increase with questions. Thus, we have $\UB\gyes^i(v) \geq \gyes^{i}(v)$ and $\UB\gno^i(v) \geq \gno^{i}(v)$. As a result, $\overline{\eg^{i}}(v) \geq \eg^{i}(v)$. 


\stitle{Algorithm}. \MTfast is a variant approach of \MTDiv in Algorithm~\ref{algo:multi} using the pruning optimization in Algorithm~\ref{algo:mtdiv+}, which calculates expected gains and identifies the vertex $q_i$ for question asking (replacing lines 7-8 of Algorithm~\ref{algo:multi}). Specifically, the algorithm first computes all upper bounds for vertices $v\in \can \setminus \yset$ and then sorts the vertices in descending order of upper bounds (lines 2-3). Next, it calculates the expected gain $\eg(v)$ and prunes disqualified vertices with an upper bound $\overline{\eg^{i}}(v) < \eg_{max}$ where $\eg_{max}$ keeps updated with the largest value of all possible expected gains (lines 4-10). Finally, it returns a vertex $q_i$ with the largest expected gain. Note that we pre-compute the \gyes and \gno of all vertices for the first question in $\T$.\xlzhu{}

%% file: tex/exp.tex
\section{Experiments}~\label{sec.exp}
In this section, we conduct extensive experiments to evaluate the performance of our proposed methods. All algorithms are implemented in C++. All the experiments are conducted on a Linux Server with AMD EPYC 7742 (2.25 GHz) and 32GB main memory.

\begin{table}[h!]
\vspace{-0.3cm}
\caption{The statistics of hierarchical tree datasets.}\label{table.data}
\vspace{-0.4cm}
\small
\centering
\scalebox{0.8}{
\begin{tabular}{|c|c|c|c|c|c|}
\toprule
Name& $|V|$& Depth& Avg\ Depth& Max\ Degree& \# Queries\\
\midrule
\revision{Image-COCO}& \revision{200}& \revision{5}& \revision{2.63}& \revision{37}& \revision{107,774}\\
ImageNet& 74,401& 19& 8.78& 391& 16,188,196\\
Yago3& 493,839& 17& 5.70& 44,538& 4,440,378\\
\bottomrule
\end{tabular}
}
\vspace{-0.4cm}
\end{table}

\begin{figure*}[t]
\vspace{-0.2cm}
\centering
{
\subfigure[\revision{Total time of $b$ questions on Image-COCO}]{
\label{fig.exp4_1}
\includegraphics[width=0.235\linewidth]{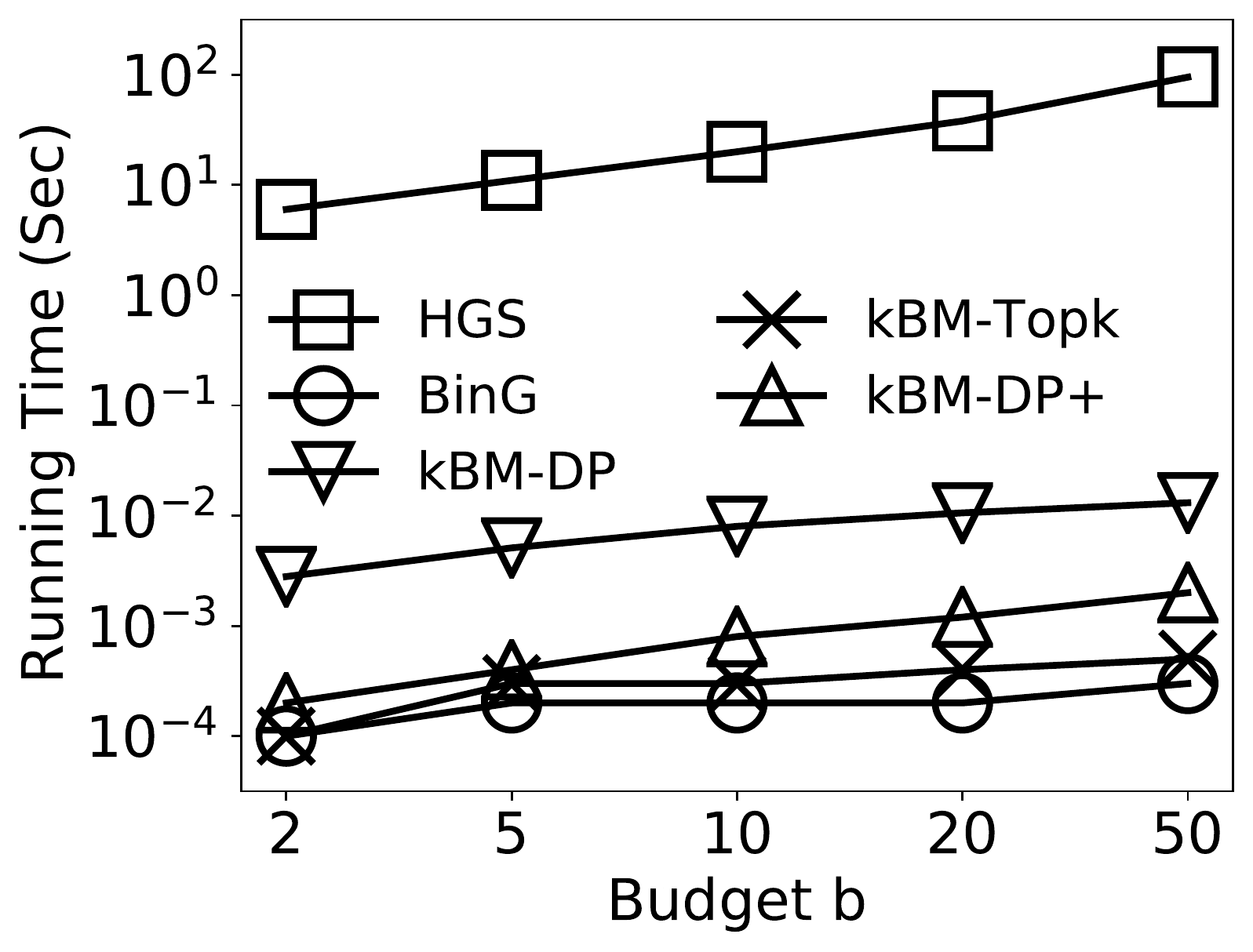} }
\subfigure[\revision{Total time of $b$ questions on ImageNet}]{
\label{fig.exp4_2}
\includegraphics[width=0.235\linewidth]{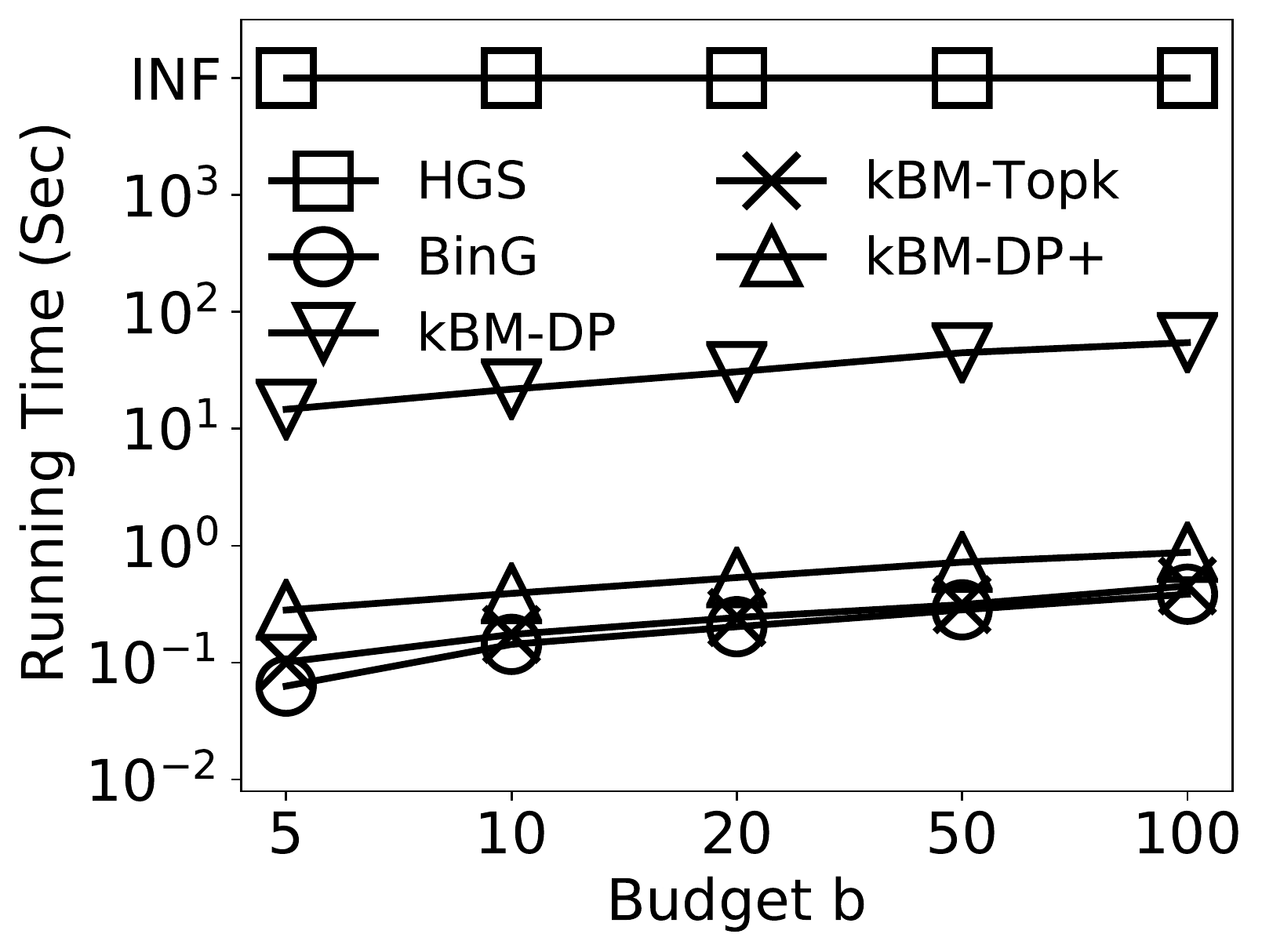} }
\subfigure[\revision{Total time of $b$ questions on Yago3-I}]{
\label{fig.exp4_1}
\includegraphics[width=0.235\linewidth]{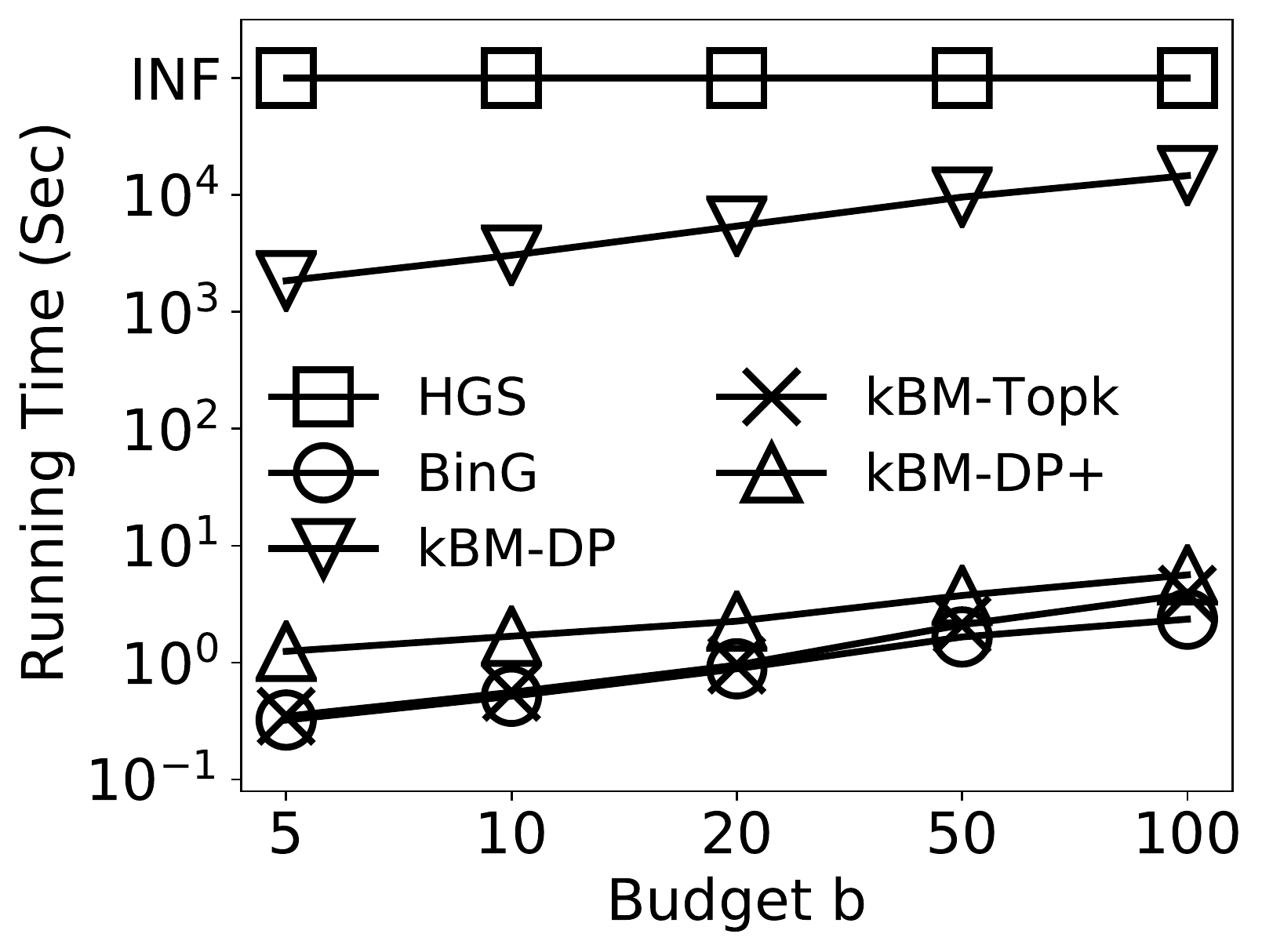} }
\subfigure[\revision{Total time of $b$ questions on Yago3-II}]{
\label{fig.exp4_2}
\includegraphics[width=0.235\linewidth]{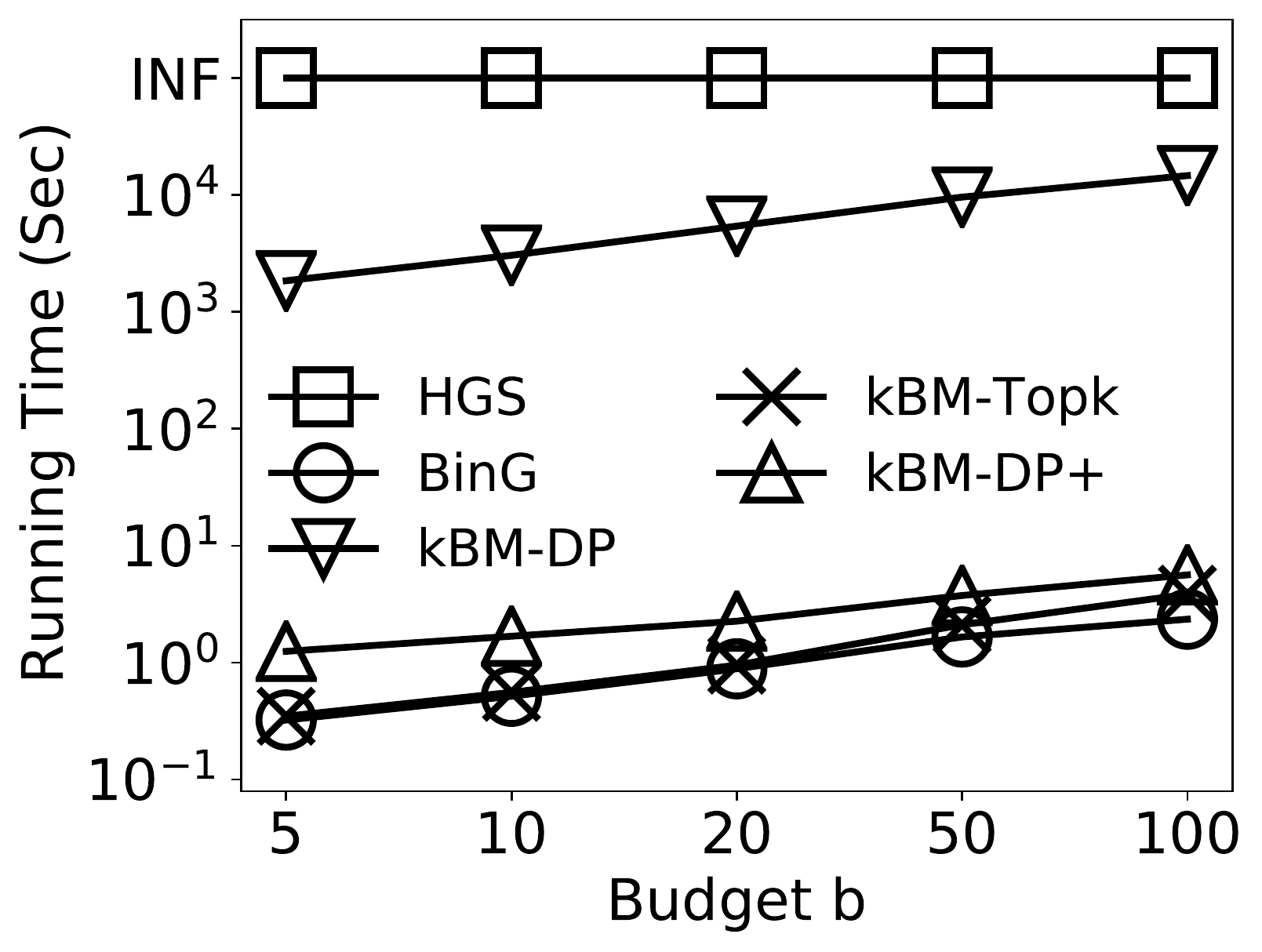} }
\quad
\subfigure[\revision{Average time per question on Image-COCO}]{
\label{fig.exp4_1}
\includegraphics[width=0.235\linewidth]{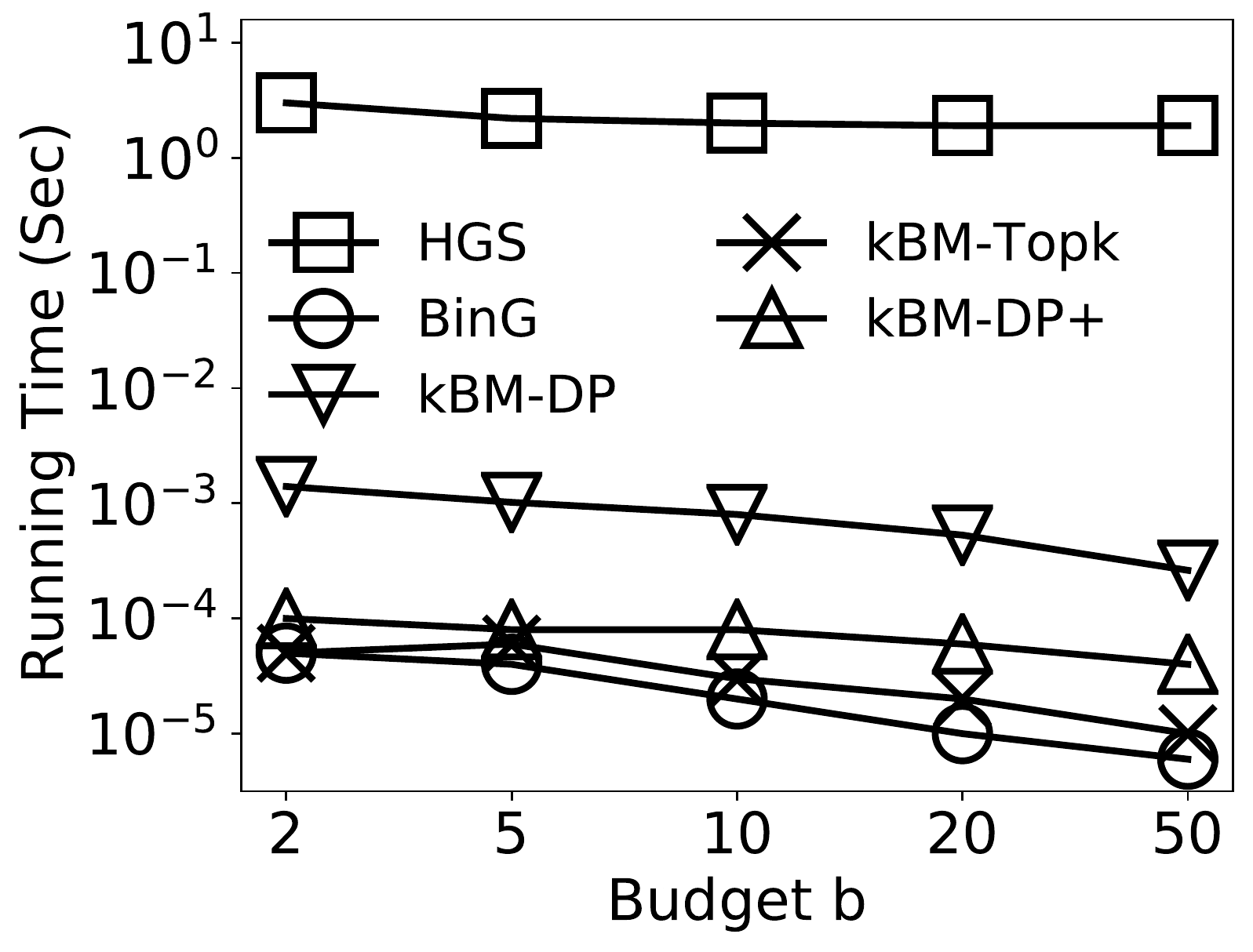} }
\subfigure[\revision{Average time per question on ImageNet}]{
\label{fig.exp4_2}
\includegraphics[width=0.235\linewidth]{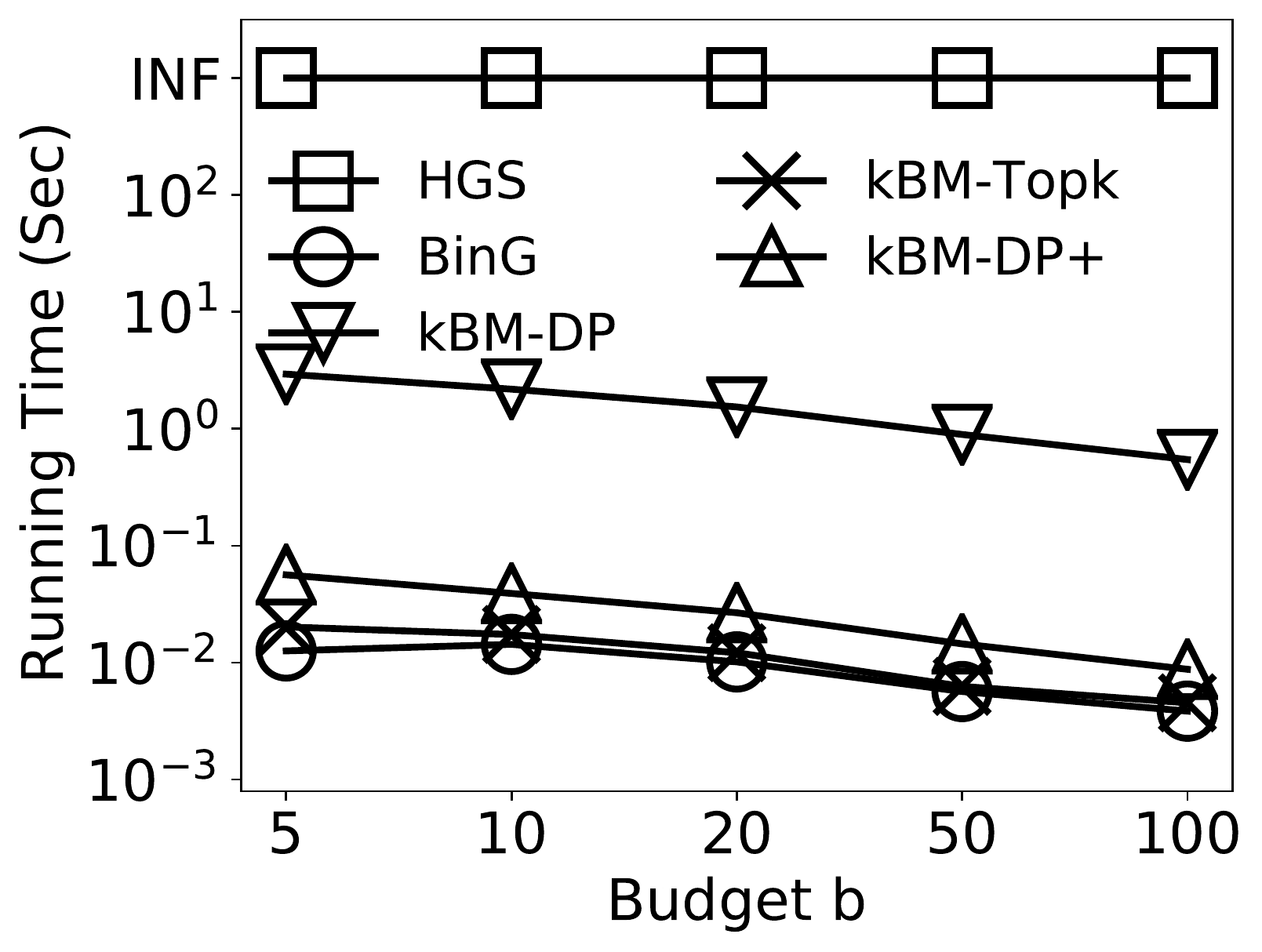} }
\subfigure[\revision{Average time per question on Yago3-I}]{
\label{fig.exp4_1}
\includegraphics[width=0.235\linewidth]{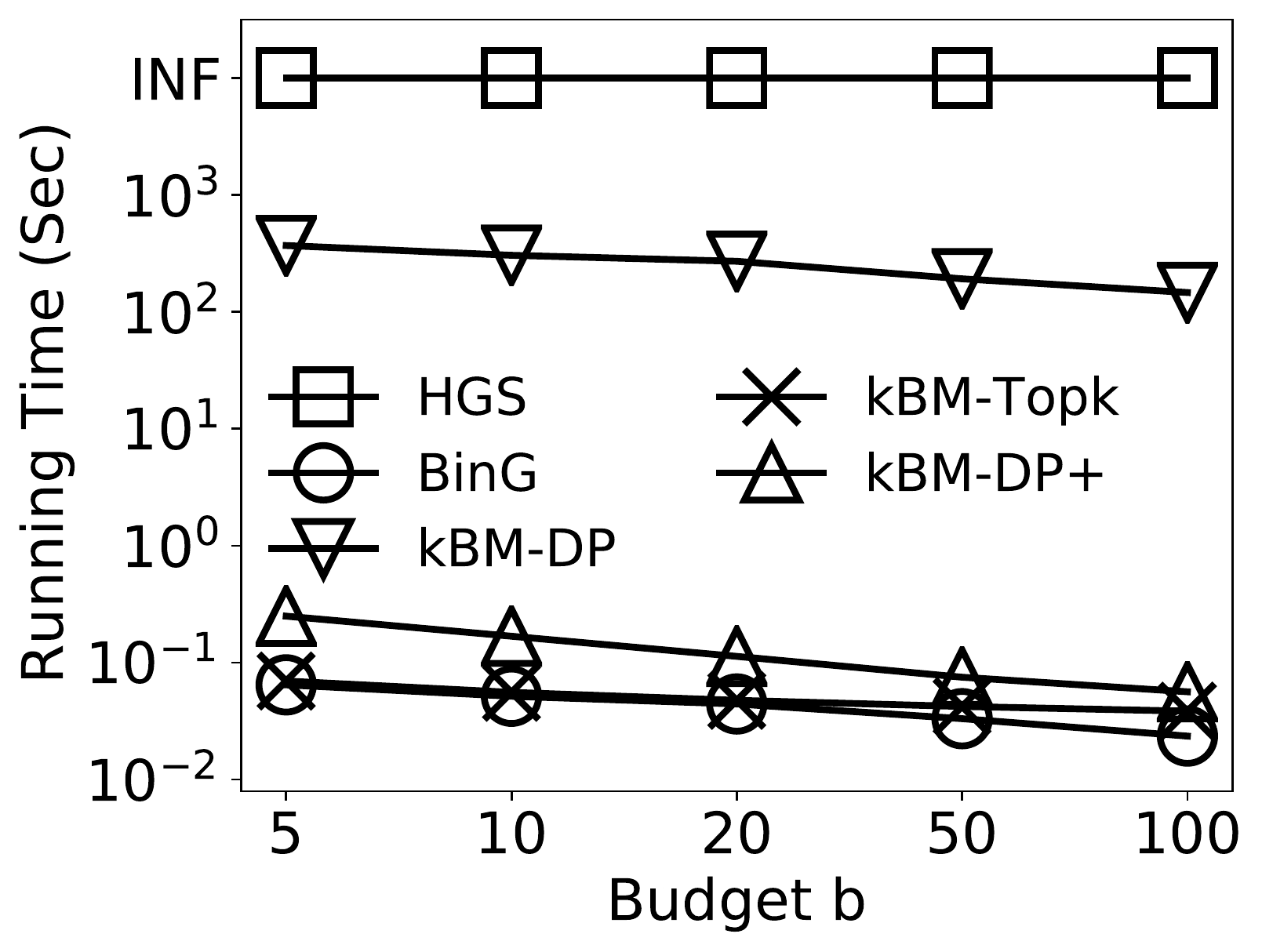} }
\subfigure[\revision{Average time per question on Yago3-II}]{
\label{fig.exp4_2}
\includegraphics[width=0.235\linewidth]{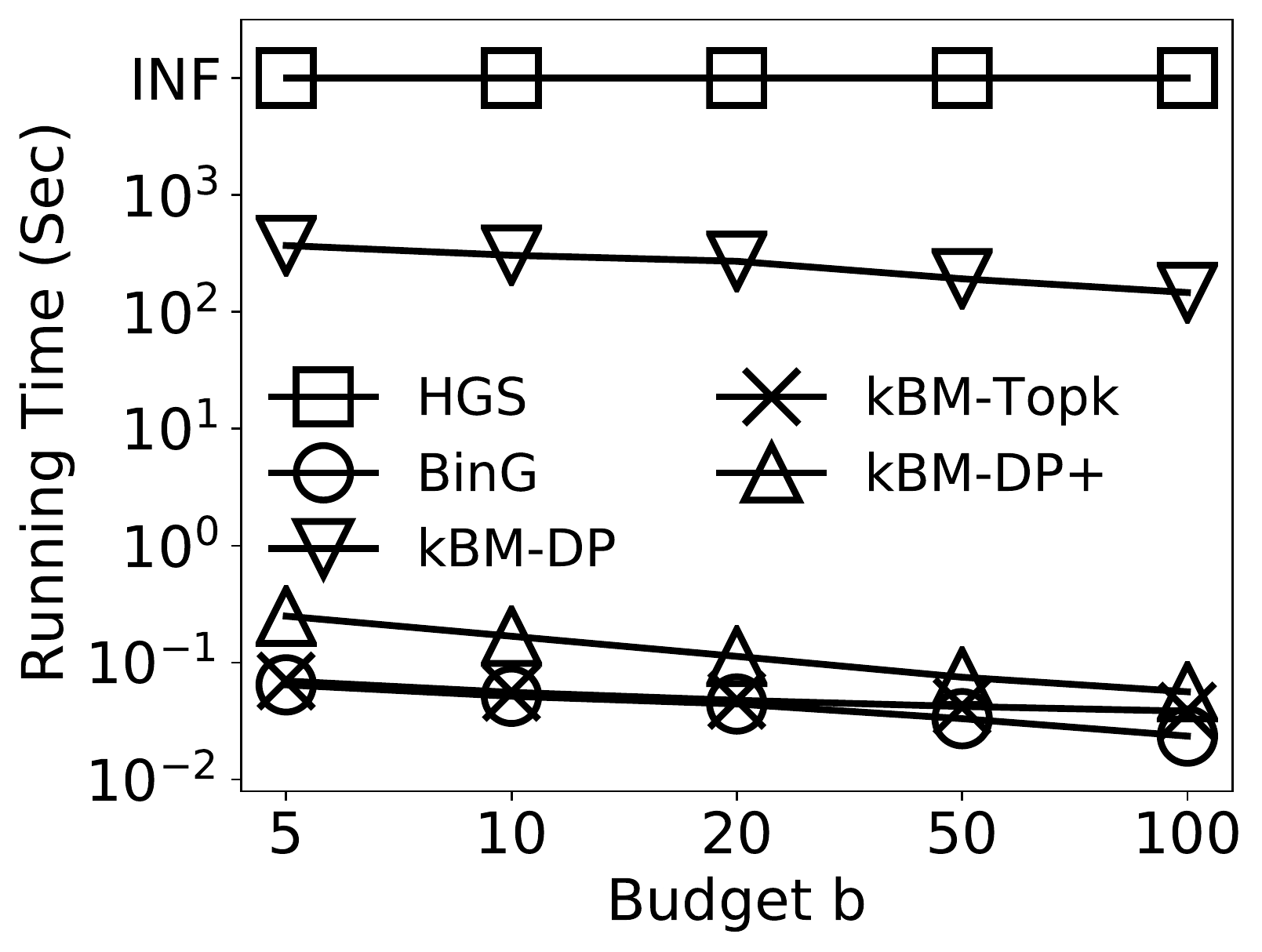} }
}
\vspace{-0.5cm}
\caption{
\revision{Efficiency evaluation of five algorithms for identifying multiple targets on all datasets.}
}
\label{fig.exp4}
\vspace{-0.6cm}
\end{figure*}


\stitle{Datasets.} We use three real datasets of hierarchical trees, whose detailed statistics are summarized in Table~\ref{table.data}. First, \emph{ImageNet}~\cite{deng2009imagenet, imagenet} is a hierarchical image dataset based on WordNet. It has 74,401 taxonomy vertices and 16 million images with ground-truth labels. 
\revision{Second, we generate a small hierarchy with 200 taxonomy vertices from COCO~\cite{lin2014microsoft} and ImageNet~\cite{deng2009imagenet}, denoted as \emph{Image-COCO}, ensuring the successful and efficient running of all tested algorithms. 
For target search on Image-COCO and ImageNet, we randomly select a set of 1,000 images with a single label and another set of 1,000 images with multiple labels for the \kw{Single}\kw{Target} and \kw{Multiple}\kw{Targets} problems, respectively.} 
Third, \emph{Yago3}~\cite{mahdisoltani2013yago3, yago} is a knowledge base from multilingual Wikipedias. We use the ontology structure yagoTaxonomy as the hierarchy for testing. It contains 493,839 taxonomy vertices, where an edge $\langle v, u \rangle$ means vertex $u$ is a ``subClassOf'' vertex $v$. Moreover, Yago3 contains 4,440,378 objects from yagoTypes, where each object may have a single label or multiple labels. 
\revision{For both the \kw{Single}\kw{Target} and \kw{Multiple}\kw{Targets} problems, we select two sets of objects with a single label and with multiple labels, respectively, using two methods. In the first method, we randomly select 1,000 labeled objects from Yago3, denoted as \emph{Yago3-I}. In the second method, we first randomly select 1,000 categories from Yago3 and then pick a random labeled object under each selected category, denoted as \emph{Yago3-II}.}

\stitle{Comparison methods.} 
\revision{We compare our  algorithms with state-of-the-art methods \HGS~\cite{parameswaran2011human}, \IGS~\cite{tao2019interactive}, and \BinG~\cite{li2020efficient}. } Specifically, 

\squishlisttight

\item \revision{ \HGS: a dynamic programming Human-GS method for identifying multiple targets with a bounded number of questions, which generates $b$ questions offline in a non-interactive setting~\cite{parameswaran2011human}. Following the algorithms in~\cite{parameswaran2011human}, 
we implement two methods, \emph{Single-Bounded} and \emph{Multi-Bounded}, for identifying a single target and multiple targets, respectively. } 

\item \IGS:  an interactive graph search algorithm for identifying a single target~\cite{tao2019interactive}. The algorithm decomposes a hierarchy into connected paths and finds the target through a series of binary searches on individual paths.

\item \BinG: \revision{a greedy algorithm for identifying a single target, which asks questions using an optimal vertex that prunes the largest number of vertices~\cite{li2020efficient}. It prunes the vertices $\can \setminus \des(u)$ for $\reach(u)=\yes$. To identify multiple targets, we implement a variant \BinG method, which only prunes 
 $\can \cap \anc(u) \setminus \{u\}$ for $\reach(u)=\yes$.}
\end{list} 

Note that both \IGS and \BinG can ask unlimited questions to identify the targets. \revision{In our problem setting, we terminate the algorithms of \IGS and \BinG after asking $b$ questions. We also evaluate and compare our proposed algorithms as follows.}
\squishlisttight
\item \STIGS: 
identifies a single target in Algorithm~\ref{algo:single}.

\item \MTDiv: a dynamic programming based method for identifying multiple targets 
in Algorithms~\ref{algo:multi} and \ref{algo:mtdiv}. 

\item \MTTopk: uses an independent penalty function to select top-$k$ vertices in Algorithm~\ref{algo:mttopk}.

\item \MTfast: uses an upper bound pruning technique to accelerate \MTDiv in Algorithm~\ref{algo:mtdiv+}.

\end{list} 

\revision{After asking $b$ questions, all algorithms return the selections from \yescs in the same way following our \Frame framework. }

\stitle{Evaluation metrics and parameter settings.} For quality evaluation, we use the penalty $\f(\sset, \tar)$ to measure the closeness between selections $\sset$ and targets $\tar$ by Def.~\ref{def.penalty}. 
For each experiment, we report the averaged penalty score of searching targets on 1,000 selected images/objects.
By default, we set the budget $b=50$, and assign $k=1$ and $k=3$ respectively for the \kw{Single}\kw{Target} and \kw{Multiple}\kw{Targets} problems. The initial probability of each vertex $v$ is set as $\pr(v) = \frac{k}{n}$ in a uniform way. \revision{We denote the running time as \emph{INF} and the penalty result as \emph{N/A}, if an algorithm cannot finish within 100 hours.}

\begin{figure*}[t]
\begin{minipage}{.6\linewidth}
\centering
\scalebox{0.75}{
\begin{tabular}{|c|cccccc|}
\toprule
Question & Label& 
$\reach(q_i)$ & depth($q_i$) & $|\can|$& $|\yset|$& $\f(\sset^*, \tar)$\\
\midrule
$q_0$& \textbf{animal}& Yes& 0& 3,998& 1& 11\\
$q_1$& \textbf{vertebrate}& Yes& 2& 3,996& 3& 7\\
$q_2$& \textbf{mammal}& Yes& 3& 3,995& 4& 6\\
$q_3$& invertebrate& No& 1& 3,219& 4& 6\\
$q_4$& \textbf{aquatic vertebrate}& Yes& 3& 3,219& 5& 5\\
$q_5$& \textbf{fish}& Yes& 4& 3,218& 6& 4\\
$q_6$& bird& No& 3& 2,347& 6& 4\\
$q_7$& bony fish& No& 5& 1,812& 6& 4\\
$q_8$& \textbf{carnivore}& Yes& 5& 1,810& 8& 2\\
$q_9$& dog& No& 7& 1,587& 8& 2\\
$q_{10}$& reptile& No& 3& 1,291& 8& 2\\
$q_{11}$& ungulate& No& 5& 988& 8& 2\\
$q_{12}$& \textbf{felid (cat family)}& Yes& 6& 987& 9& 1\\
\bottomrule
\end{tabular}
}
\end{minipage}
\hspace{-1.0cm}
\begin{minipage}{.4\linewidth}
\centering
\includegraphics[width=0.65\linewidth]{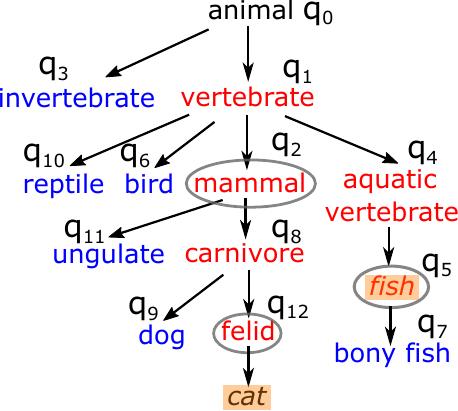}
\end{minipage}
\vspace{-0.4cm}
\caption{Case study on the ``animal'' hierarchy in ImageNet dataset. The targets are $\tar =$ \{``fish'', ``cat''\}. Here, $b = 12$ and $k = 3$. The selection set of our algorithm is $\sset = $\{``fish'', ``mammal'', ``felid''\}, in which ``felid'' is the parent of ``cat''. The penalty score is $\f(\sset, \tar) = 1$.}
\label{fig.exp.case_study}
\vspace{-0.5cm}
\end{figure*}

\stitle{EXP-1: Quality evaluation of  the \kw{Single}\kw{Target} problem.} \revision{
Table~\ref{table.exp1} shows the penalty results of four methods \HGS, \IGS, \BinG, and \STIGS for identifying a single target. For each dataset, we test five different budgets of $b$, varying from $5$ to $100$. 
The smaller the penalty scores, the closer the selections to the hidden targets.
All methods get lower penalty scores with increased budget $b$ as more questions are asked to obtain better selections. Our method \STIGS achieves the best performance in all tests, except for one case of $b = 5$ on Yago3-I. 
In particular, it outperforms \HGS by 23\%--1,253\%. 
While \BinG has a competitive performance with \STIGS, it is much worse than our methods in the more challenging \kw{MultipleTargets} problem as will be shown in EXP-2.
}

\stitle{EXP-2: Quality evaluation of the   \kw{Multiple}\kw{Targets} problem.} \revision{We evaluate four methods \MTHGS, \MTBinG, \MTTopk, and \MTfast for identifying multiple targets.
Figures~\ref{fig.exp2}(a)-\ref{fig.exp2}(d) and Figures~\ref{fig.exp2}(e)-\ref{fig.exp2}(h) report the penalty results on all datasets by varying budget $b$ and selection size $k$, respectively. Several observations are made.  
First, \MTHGS has the largest penalty scores 
on the small dataset Image-COCO as shown in Figures~\ref{fig.exp2}(a) and~\ref{fig.exp2}(e).
On the three large datasets in Figures~\ref{fig.exp2}(b)-\ref{fig.exp2}(d) and \ref{fig.exp2}(f)-\ref{fig.exp2}(h), \MTHGS fails to finish within 100 hours, due to its high time complexity. 
Second, compared with \MTBinG, our methods \MTfast and \MTTopk get smaller penalty scores by achieving an average of 2.1{x} better results. 
The main reason is that \MTBinG tends to ask questions on the vertices at the bottom levels, which likely gets a $\no$ answer with little gain of reducing target penalties. In contrast, our methods \MTTopk and \MTfast aim at asking questions on the vertices with the largest expected gains based on the potential target distribution, thereby achieving a better performance. Moreover, with  increased budget $b$ and target number $k$,  \MTTopk and \MTfast get an even better performance with lower  penalty scores. Finally, between \MTfast and \MTTopk, \MTfast incurs less penalties because it has a better gain function for identifying diverse selections.
}
\begin{figure}[t!]
\centering
{
\subfigure[Yago3-I]{
\label{fig.exp3_1}
\includegraphics[width=0.4\linewidth]{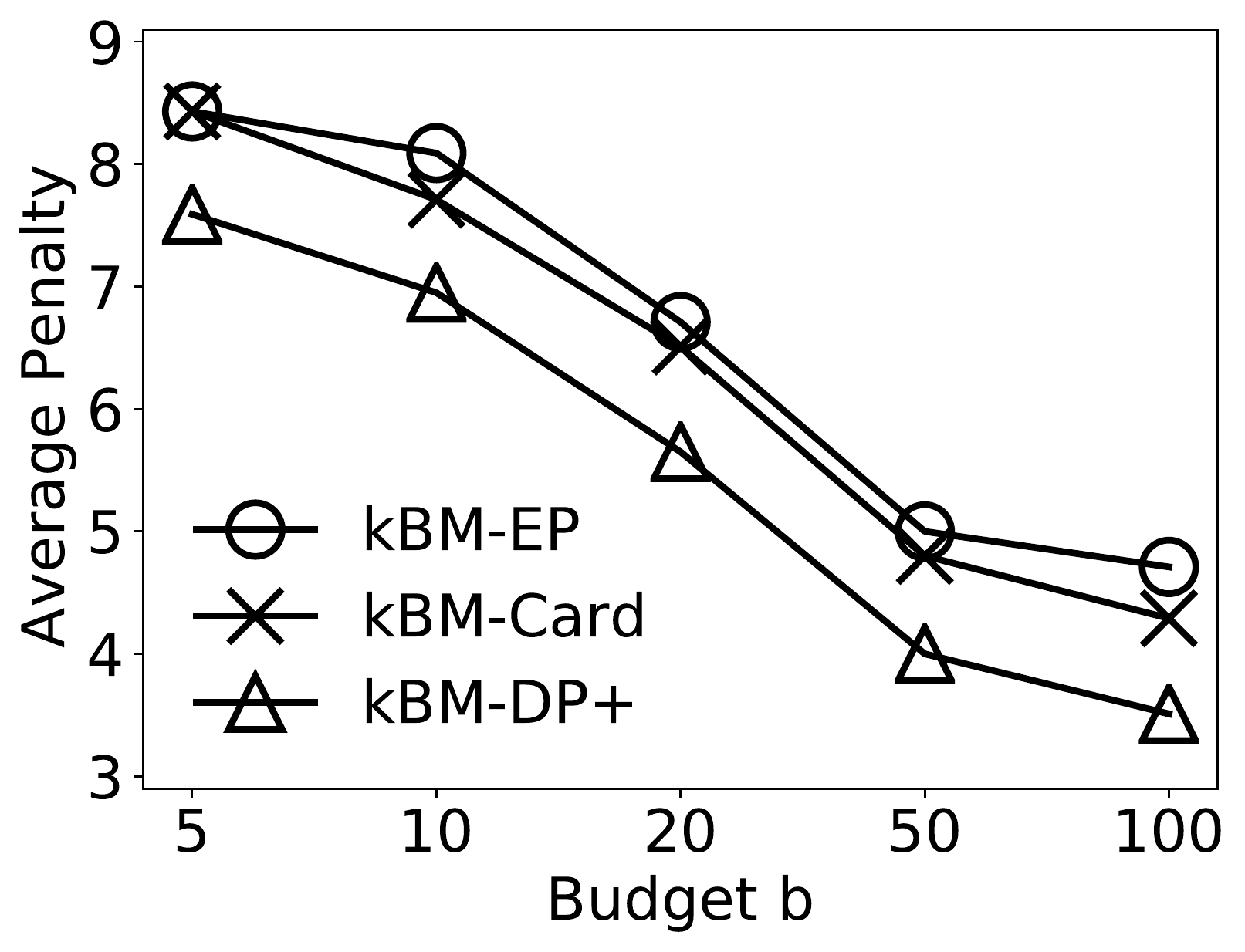} }
\subfigure[Yago3-I]{
\label{fig.exp3_2}
\includegraphics[width=0.4\linewidth]{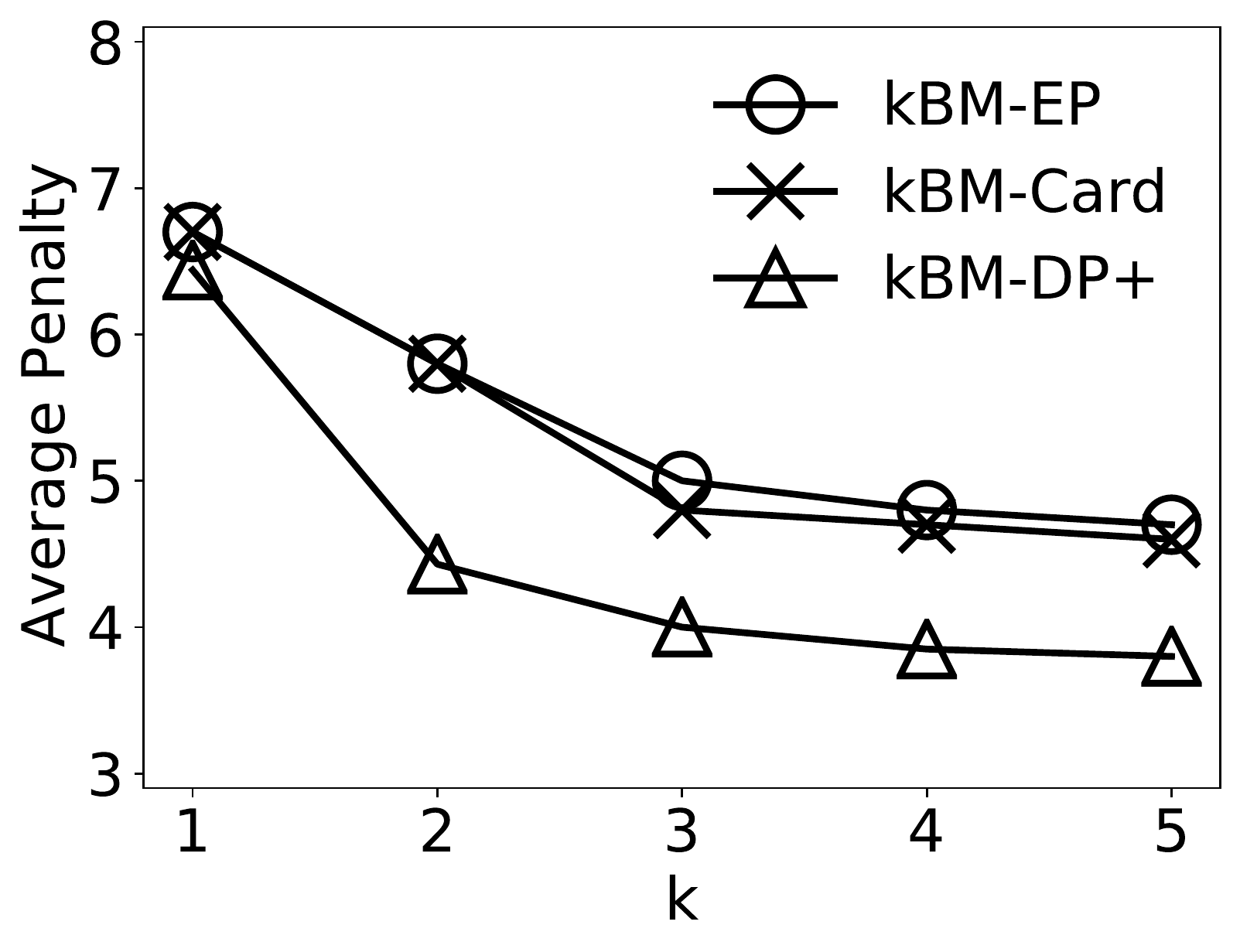} }
}
\vspace{-0.5cm}
\caption{The quality evaluation of \mtigs framework on Yago3-I,  w.r.t. the probability distribution and gain function. 
}
\label{fig.exp3}
\vspace{-0.6cm}
\end{figure}

\stitle{EXP-3: Quality evaluation of kBM-IGS framework.} 
We evaluate the effectiveness of the \mtigs framework using different probability and gain score functions. We implement two variant methods of \mtigs as \MTEP and \MTND. Both \MTEP and \MTND follow the same framework in Algorithm~\ref{algo:framework}. However, \MTEP sets a uniform value of $0.5$ for all vertices' probability scores of $\pyes$ and $\pno$. For two gain score functions $\gyes$ and $\gno$, \MTND counts the number of potential targets in a vertex's  descendants, without considering distance. Figure~\ref{fig.exp3_1} and Figure~\ref{fig.exp3_2} report the results of \mtigs, \MTEP, and \MTND varied by budget $b$ and selection size $k$ respectively. \mtigs wins over both \MTEP and \MTND, demonstrating the effectiveness of the probability and gain score functions used in our \mtigs framework.




\stitle{EXP-4: Efficiency evaluation.} \revision{We evaluate the efficiency of generating questions for multiple targets on all datasets. 
Figure~\ref{fig.exp4} shows the running time results of five different algorithms \MTHGS, \MTBinG, \MTDiv, \MTTopk, and  \MTfast.
Note that \MTHGS fails to finish on ImageNet, Yago3-I, and Yago3-II. \MTfast offline pre-calculates the \UB\gyes and \UB\gno of the first question, which asks the same question for any targets. 
\MTTopk is consistently faster (2.8\textbf{x} on average) than \MTfast for different settings of parameter $b$ in Figure~\ref{fig.exp4}(a)-(h).
Because of the low efficiency of \MTDiv, we only run $10$ cases for it. As shown in Figure~\ref{fig.exp4}(a)-(d), all methods take more time with increased budget $b$. 
\MTTopk and \MTBinG generate questions fastest by adopting the top-$k$ penalty function. Furthermore, all methods take less average running time for each question with increased budget $b$ as shown in Figure~\ref{fig.exp4}(e)-(h).}

\begin{figure}[t]

\centering
{
\subfigure[The number of potential targets]{
\label{fig.exp-purning1}
\includegraphics[width=0.45\linewidth]{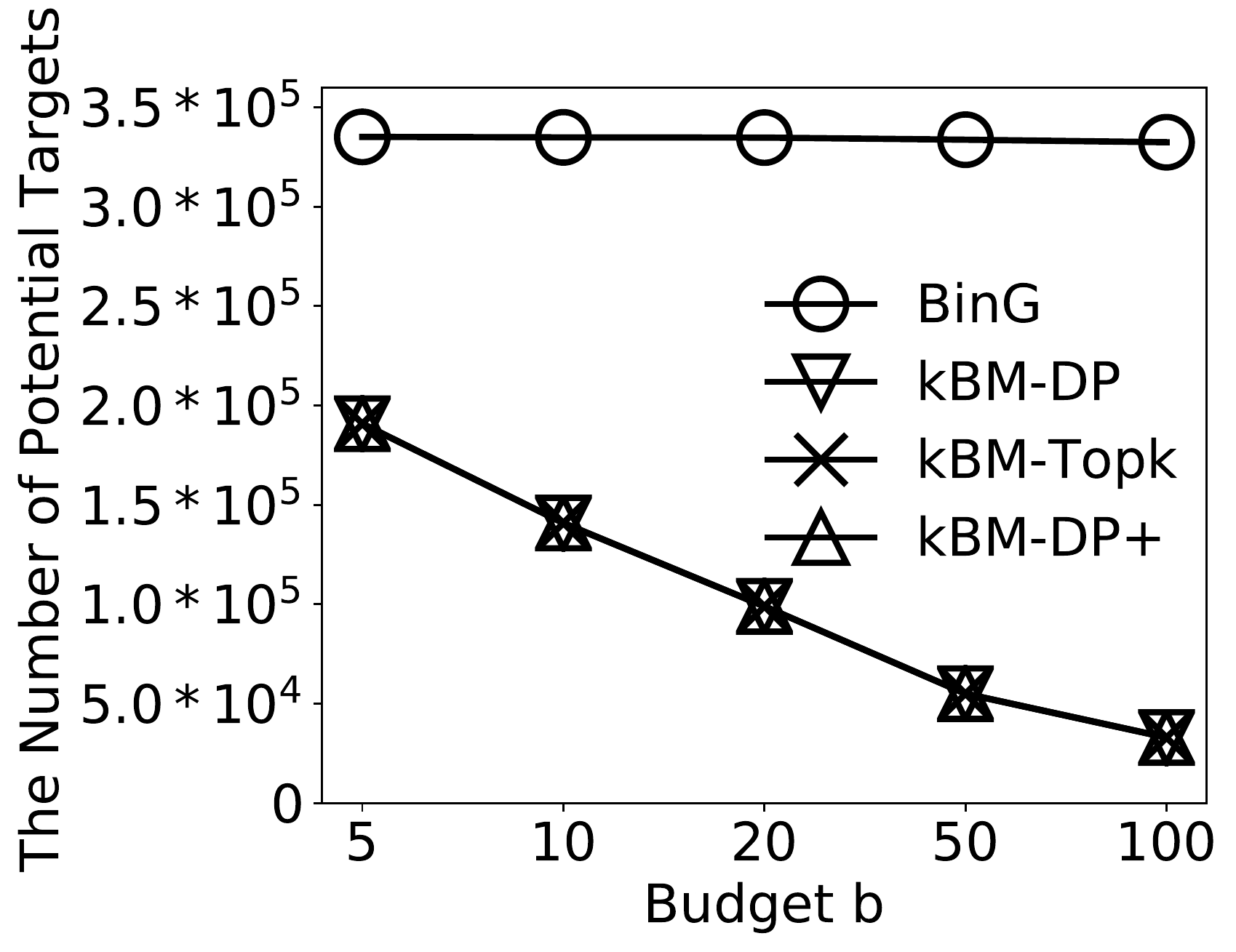} }
\subfigure[The number of  gain calculations]{
\label{fig.exp-purning2}
\includegraphics[width=0.4\linewidth]{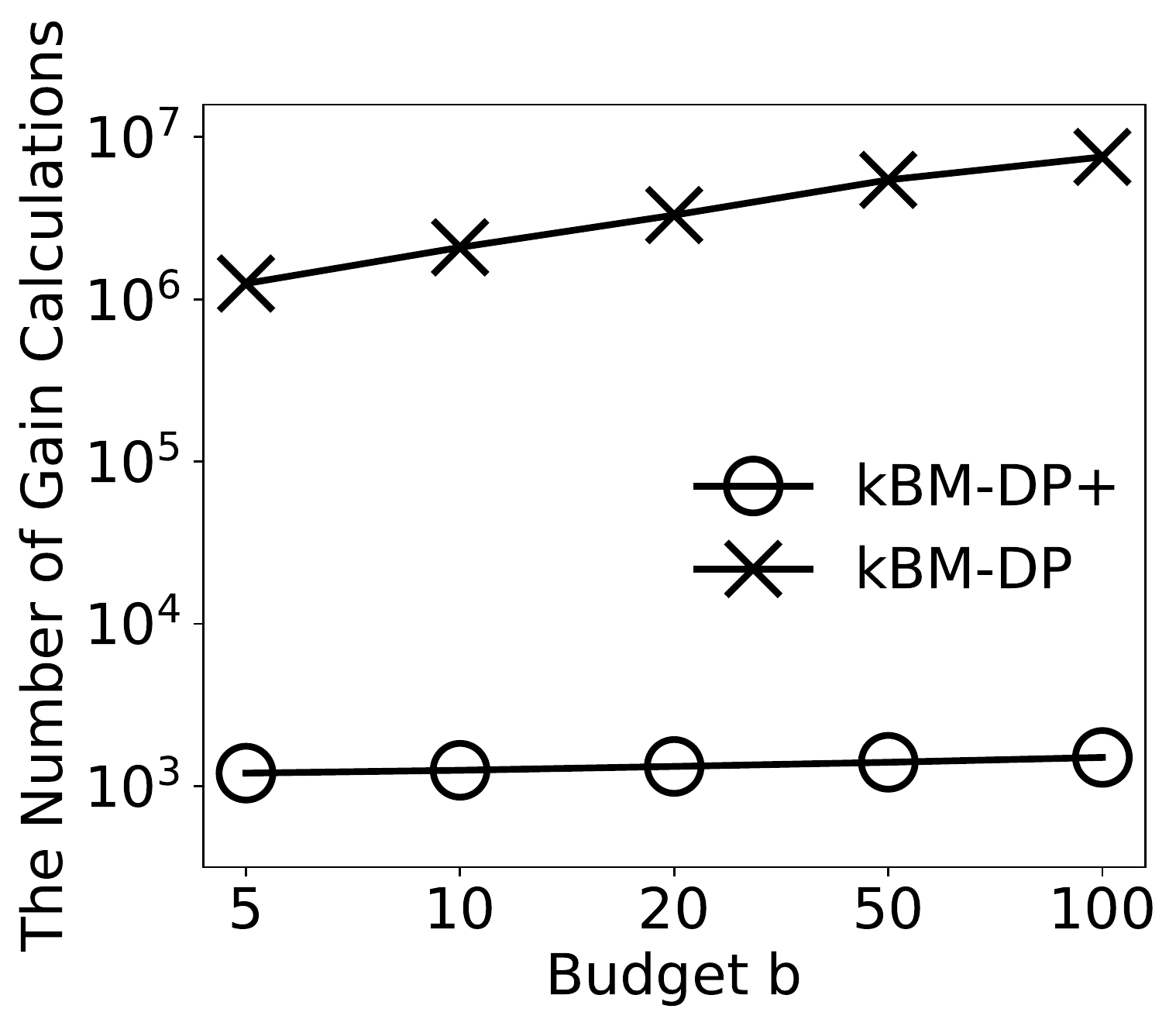} }
}
\vspace{-0.5cm}
\caption{\revision{Pruning ability evaluation on Yago3-I.}}
\label{fig.purning}
\vspace{-0.6cm}
\end{figure}

\stitle{EXP-5: Pruning ability evaluation.} \revision{We conduct a pruning ability evaluation and report the results in Figures~\ref{fig.exp-purning1},~\ref{fig.exp-purning2}. 
Figure~\ref{fig.exp-purning1} shows the number of potential targets by our methods \MTDiv, \MTTopk, \MTfast and competitive method \MTBinG. Our methods \MTDiv, \MTTopk, and \MTfast consistently outperform \MTBinG with the increased number of budgets $b$. 
However, in Figure~\ref{fig.exp4}, \MTfast runs much faster than \MTDiv, due to its efficient expected gain calculation as shown in Figure~\ref{fig.exp-purning2}. Compared with \MTDiv, \MTfast takes three orders of magnitude 
less calculations of expected gain, which validates the pruning optimization strategies in Section~\ref{sec.dp+}.
}

\begin{figure}[t]

\centering
{
\subfigure[ImageNet]{
\label{fig.exp5_1}
\includegraphics[width=0.45\linewidth]{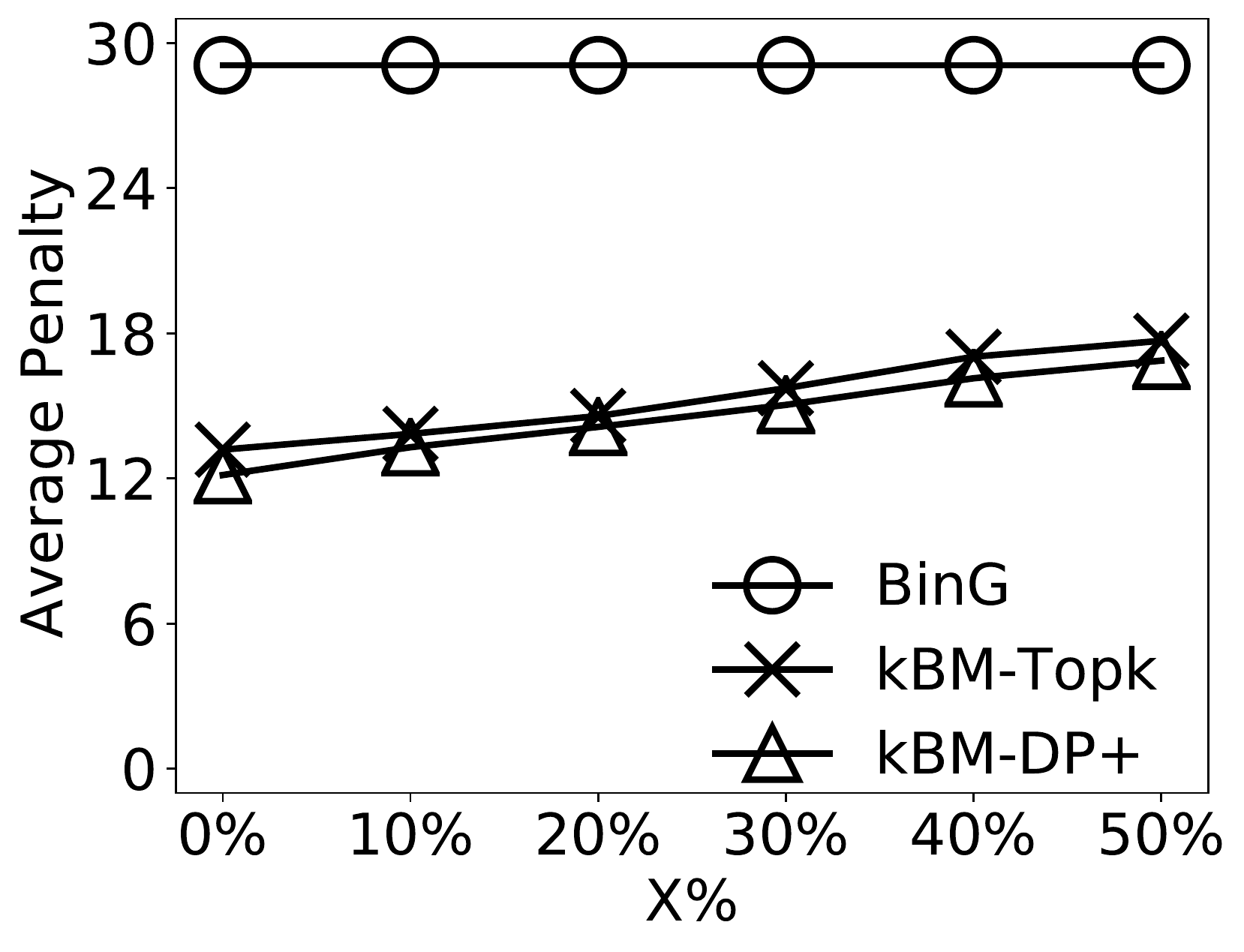} }
\subfigure[Yago-I]{
\label{fig.exp5_2}
\includegraphics[width=0.45\linewidth]{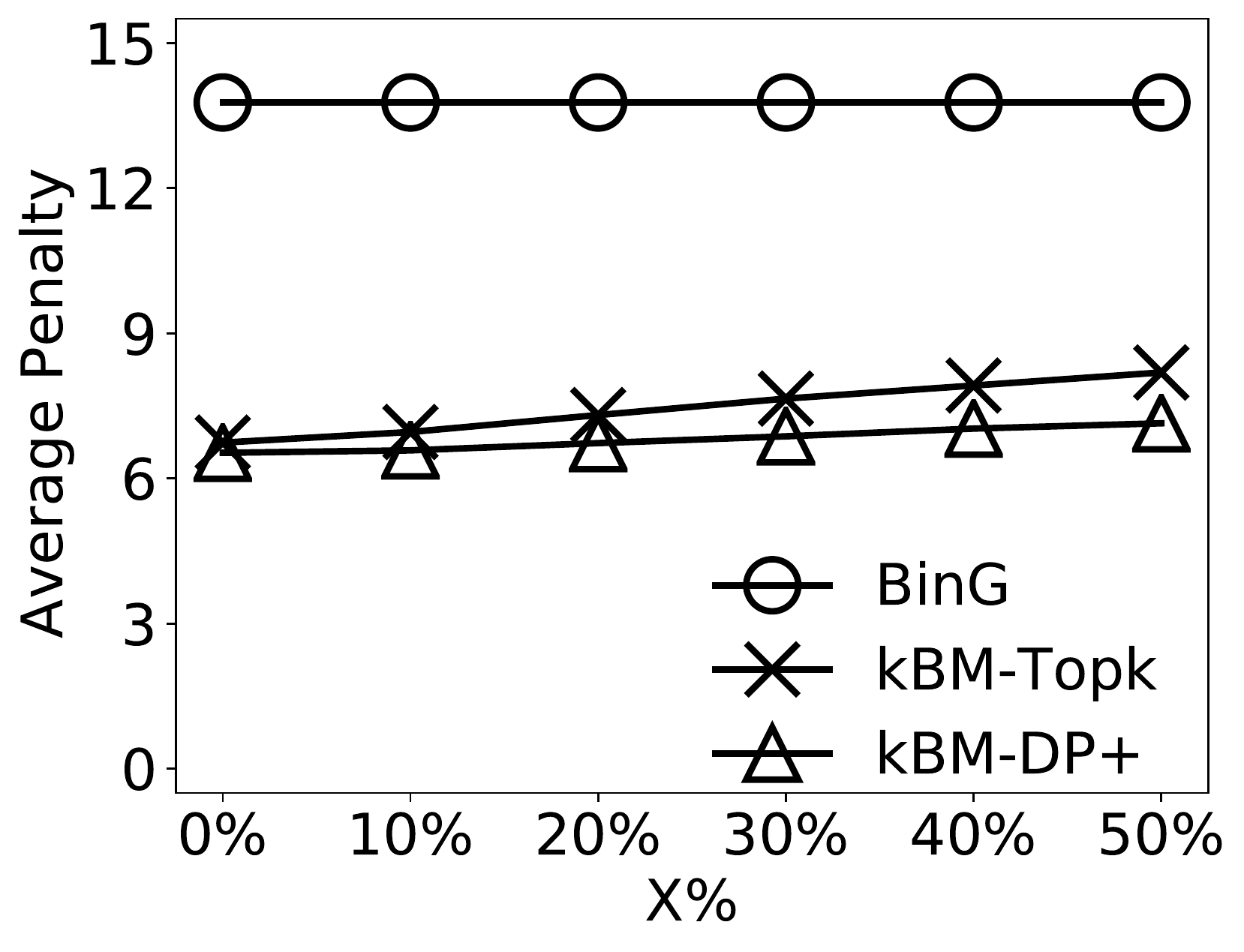} }
}
\vspace{-0.5cm}
\caption{\revision{Quality evaluation with wrong answers.}}
\label{fig.exp5}
\vspace{-0.6cm}
\end{figure}

\stitle{EXP-6: Evaluation of the quality of \mtigs scheme incurred when wrong answers are received from the crowd.} \revision{
We conduct a quality evaluation of our methods where human mistakes are not eliminated and the workers give wrong answers. For each dataset, we randomly select $X\%$ objects out of 1,000 objects and treat them as difficult objects. We vary $X \in [0, 50]$ on the ImageNet and Yago3-I datasets. 
For each question that involves a difficult object, the workers have a probability of giving a wrong answer, denoted as $p$. In the experiment, we set \emph{the wrong probability $p = 10\%$} and \emph{budget $b = 50$}. Figure~\ref{fig.exp5} shows the penalty results when varying the percentage of difficult objects.
As can be seen, the quality performances of \MTTopk and \MTfast are only slightly degraded with the increasing percentage of difficult objects, demonstrating their resilience to wrong answers. Moreover, our methods \MTTopk and \MTfast still win BinG by at least 40\%,
 even with wrong answers.
}

\stitle{EXP-7: Case study of image categorization.} We conduct a case study of interactive graph search to identify multiple targets on ImageNet. We extract an `animal' sub-hierarchy of ImageNet, which contains nearly $4,000$ labels. We use an image shown in Figure~\ref{fig.fish} with these targets $\tar =$ \{``fish'', ``cat''\}. We apply the \MTfast method with a budget $b = 12$ and selection size $k = 3$. The left table in Figure~\ref{fig.exp.case_study} shows the detailed process and statistics of all interactive questions by \MTfast. 
For each question vertex $q_i$, 
we report the label of $q_i$, the answer $\reach(q_i)$, the depth of $q_i$ in $\T$, 
$|\can|$, $|\yset|$, and 
the penalty $\f(S^*, \tar)$.
By default, $q_0=r$. 
 We also show the questioned taxonomies in a simplified hierarchy on the right side of Figure~\ref{fig.exp.case_study}. 
   The red taxonomies get a $\yes$ answer and the blue taxonomies get a $\no$ answer. 
   The interactive process clearly shows that our questions approach the targets quickly in a top-down manner within 12 questions, which achieves a very small penalty of 1 between selections and targets. 
   Finally, \MTfast identifies the selections $\sset =$ \{``fish'', ``mammal'', ``felid''\}. Note that ``felid'' means the cat family, which is the parent of the target ``cat''. The reason for selection ``mammal'' is   because some subclass labels of ``mammal'', such as ``primate'' and ``rodent'', are potential targets.